\documentclass[11pt]{article}

\usepackage{geometry}
\usepackage[ruled]{algorithm2e} 

\SetAlFnt{\small}
\SetAlCapFnt{\small}
\SetAlCapNameFnt{\small}
\SetAlCapHSkip{0pt}
\IncMargin{-\parindent}

\usepackage{tikz}
\usetikzlibrary{decorations.pathreplacing}

\usepackage{bbold}
\usepackage{booktabs}
\usepackage{amsmath, amsthm, amsfonts,dsfont}
\usepackage{indentfirst}
\usepackage[capitalize]{cleveref}
\usepackage{enumitem}
\usepackage{wrapfig}
\usepackage{subcaption}
\usepackage{xspace,wrapfig}
\usepackage{thmtools}
\usepackage{natbib}
\usepackage{xcolor}
\usepackage{color-edits}

\newtheorem{theorem}{Theorem}[section]
\newtheorem{corollary}{Corollary}[section]
\newtheorem{lemma}[theorem]{Lemma}

\newtheorem{proposition}[theorem]{Proposition}

\theoremstyle{definition}
\newtheorem{definition}{Definition}[section]
\newtheorem{example}{Example}[section]

\newenvironment{itemize*}
{\begin{itemize}
    \setlength{\itemsep}{1pt}
    \setlength{\parskip}{1pt}}
{\end{itemize}}
\newenvironment{enumerate*}
{\begin{enumerate}
    \setlength{\itemsep}{1pt}
    \setlength{\parskip}{1pt}}
{\end{enumerate}}

\DeclareMathOperator*{\E}{\mathbf{E}}
\let\Pr\relax
\DeclareMathOperator*{\Pr}{\mathbf{Pr}}

\newcommand{\mathsc}[1]{{\normalfont\textsc{#1}}}

\newcommand{\opt}{\mathsc{Opt}}
\newcommand{\cs}{\mathsc{CS}} 

\newcommand{\BBM}{\mathsc{BBM}}
\newcommand{\rank}{\mathsc{Rank}}
\newcommand{\TB}{\mathsc{Rank}}
\newcommand{\cut}{\mathsc{Trunc}}
\newcommand{\ind}{\mathbb{1}}

\newcommand{\alg}{\mathsc{Alg}}
\newcommand{\core}{\mathsc{Core}}
\newcommand{\tail}{\mathsc{Tail}}

\newcommand{\M}{\mathcal{M}} 
\newcommand{\MS}{\mathcal{M}_{\mathcal{S}}} 
\newcommand{\Rm}{\R^{M}}

\newcommand{\vc}{v^{\circ}}

\renewcommand{\S}{\mathcal{S}} 

\newcommand{\D}{\mathcal{D}}
\newcommand{\W}{\mathcal{W}}
\newcommand{\R}{\mathcal{R}}
\newcommand{\K}{\mathcal{K}}
\newcommand{\Z}{\mathcal{Z}}

\newcommand{\latin}[1]{#1}

\begin{document}

\title{The Limits of an Information Intermediary in Auction Design}

\makeatletter
\def\@fnsymbol#1{\ensuremath{\ifcase#1\or \dagger\or \ddagger\or
   \mathsection\or \mathparagraph\or \|\or **\or \dagger\dagger
   \or \ddagger\ddagger \else\@ctrerr\fi}}
\makeatother
\newcommand*\samethanks[1][\value{footnote}]{\footnotemark[#1]}

\author{Reza Alijani\thanks{Department of Computer Science, Duke University. Email: \texttt{\{alijani,kamesh,knwang\}@cs.duke.edu}} \and Siddhartha Banerjee\thanks{School of Operations Research and Information Engineering, Cornell University. Email: \texttt{sbanerjee@cornell.edu}.} \and Kamesh Munagala\samethanks[1] \and Kangning Wang\samethanks[1]}

\date{}

\maketitle

\begin{abstract}
We study the limits of an information intermediary in the classical Bayesian auction, where a revenue-maximizing seller sells one item to $n$ buyers with independent private values. In addition, we have an intermediary who knows the buyers' private values, and can map these to a public signal so as to increase consumer surplus. This model generalizes the single-buyer setting proposed by Bergemann, Brooks, and Morris, who present a signaling scheme that raises the optimal consumer surplus, by guaranteeing that the item is always sold and the seller gets the same revenue as without signaling. Our work aims to understand how this result ports to the setting with multiple buyers.

We likewise define the benchmark for the optimal consumer surplus: one where the auction is efficient (i.e., the item is always sold to the highest-valued buyer) and the revenue of the seller is unchanged. We show that no signaling scheme can guarantee this benchmark even for $n=2$ buyers with $2$-point valuation distributions. Indeed, no signaling scheme can be efficient while preserving any non-trivial fraction of the original consumer surplus, and no signaling scheme can guarantee consumer surplus better than a factor of $\frac{1}{2}$ compared to the benchmark. These impossibility results are existential (beyond computational), and provide a sharp separation between the single and multi-buyer settings.

In light of this impossibility, we develop signaling schemes with good approximation guarantees to the benchmark. Our main technical result is an $O(1)$-approximation for i.i.d. regular buyers, via signaling schemes that are conceptually simple and computable in polynomial time. We also present an extension to the case of general independent distributions.
\end{abstract}

\section{Introduction}
Consider a seller selling an item to a buyer, whose private value $V$ is drawn from some known distribution~$\D$. 
The overall \emph{social welfare} is maximized when the seller sells the item for $\$ 0$, assuming the seller has no cost for the item. In contrast, to maximize the (average) revenue, the seller's optimal strategy is to sell at a revenue-maximizing price, which may lead to welfare loss due to the item going unsold.

More generally, in a single-item Bayesian auction with $n$ buyers with independent private valuations, a welfare-optimal mechanism is the second-price (or VCG) auction, which always gives the item to the highest-valued buyer. In contrast, even when the buyers have i.i.d. regular valuations, the revenue-optimal mechanism was shown by \citet{myerson} to be a second-price auction with a \emph{reserve price}; this may lead to the item going unsold. The situation is more complex for non-regular distributions, and/or non-i.i.d.~buyers, where the revenue-optimal mechanism may in addition sell the item to a buyer with lower value than the highest, leading to additional welfare loss. 
We visualize this via a \emph{revenue-CS trade-off} diagram (\cref{fig1b}), where, for different mechanisms and value distributions, we plot expected consumer-surplus (\latin{i.e.}, value minus payment), denoted $\cs$, versus expected seller-revenue, denoted by $\R$. Any welfare-maximizing mechanism including VCG (point $V$) lies on the $\R + \cs = \W^*$ line. In contrast, Myerson's mechanism (point $M$) has revenue $\R^M$ greater than that under VCG, but can also lie below the maximum-welfare line.

\paragraph{Information Intermediary.} Now consider the same setting, but with an additional \emph{information intermediary}: a third-party who knows the {\em true} buyer values $\vec{V}=(V_1,V_2,\ldots,V_n)$ and can provide a ``signal''  or side-information to the seller and the buyers. 
Both the signal and signaling scheme are common knowledge to all agents (buyers and seller), who can thus use Bayes' rule to update the prior over valuations given the signal.  The signal ``re-shapes'' the joint prior over the buyer valuations in a \emph{Bayes-plausible} manner (\latin{i.e.}, such that the posterior averaged over signals equals the prior). Though the intermediary can modulate information, it does not control the mechanism, which still resides with the seller. Such a setting is motivated by ad exchanges, where the platform (or intermediary) acts only as a clearinghouse, and does not itself run a mechanism. Therefore, given the signal, the seller then proposes the revenue-maximizing mechanism, and buyers bid optimally, under the posterior distribution. We illustrate this in~\cref{fig1a}.

Formally, consider a setting where $n$ buyers have independent private valuations $\vec{V}$ drawn from a distribution $\D=\D_1\times\D_2\times\cdots\times\D_n$. The valuations $\vec{V}$ are known to the intermediary, who maps them to a signal $\sigma$ via a public signaling scheme $\Z$. Given $\sigma$, all agents compute the posterior $\S$ over buyer values; note these can now be correlated.  The seller then proposes a mechanism $\MS$ (comprising allocation and payment rules) which maximizes its expected revenue assuming buyers act in a manner which is \latin{ex-post} incentive-compatible (IC) and interim individually-rational (IR) given $\S$. If $\sigma$ is such that $\S=\D$, then $\MS$ is Myerson's auction (point $M$ in~\cref{fig1b}); on the other hand, if the signal fully reveals $\vec{V}$, then the seller can extract full surplus (\latin{i.e.}, get revenue $\W^*$, point $A$ in~\cref{fig1b}). Moreover, the seller gets revenue at least $\Rm$ under any signaling scheme, as she can always ignore the signal (see~\cref{sec:prelim}). Thus any signaling scheme $\Z$ must give a point in the shaded triangle with consumer surplus $\cs(\Z)$ and revenue $\R(\Z)$, and the maximum possible surplus $\opt$ is achieved at point $O$ in~\cref{fig1b}.
Now we can ask:
\begin{quote}
\it What revenue-CS trade-offs can an information intermediary achieve via signaling? More specifically, what is the  maximum possible consumer surplus that is achievable? 
\end{quote}

In the single-buyer case, the seminal work of \citet*{bergemann2015limits} completely answer these questions by showing that the entire shaded region is always achievable. In particular, the point $O$ is met by a simple signaling scheme where the revenue is {\em exactly} $\Rm$, and the item is \emph{always} sold thus the mechanism is efficient.

In this work we study the effectiveness of an information intermediary in a multi-buyer (\latin{i.e.}, $n \geq 2$ buyers) Bayesian auction. In brief, we expose a sharp separation between the single and multi-buyer settings, as in the latter, no signaling scheme can guarantee more than a constant fraction of the optimal consumer surplus ($\opt$ in~\cref{fig1b}). 
On the positive side, 
we obtain a novel yet simple signaling scheme with strong approximation guarantees for a wide range of settings.

While our main focus is on theoretical results, our work has broader practical relevance.
Consider an agency like the FCC with privileged information about bidders in a spectrum auction, or a bid optimizer working for multiple competing clients in an ad-exchange.
These intermediaries have private information about the buyers, and can selectively release it to influence the auction. For instance, in an ad exchange, the platform running the exchange (intermediary) typically uses machine learning and advertiser features to infer true valuations. However, the pricing rules are decided by the publishers (seller) and not the exchange. For its own long term viability, the platform clearly has incentives to make both parties -- publisher and advertisers -- as happy as possible, and would therefore like to release information selectively to the publisher in order to maximize advertiser happiness (consumer surplus) while keeping publisher happiness (revenue) at least what it is without its presence. In effect, we use the alternate view of the market segmentation problem in~\citep{bergemann2015limits} as a special case of a signaling problem where a more informed intermediary works for the benefit of the buyers.

Our work also fits in a broader space of multi-criteria optimization  where a third-party platform or government agency can release information about agents to a principal in charge of an activity such as admissions or hiring, so as to trade-off the principal's objective such as maximizing quality of hire, with a societal objective such as fairness or diversity. 

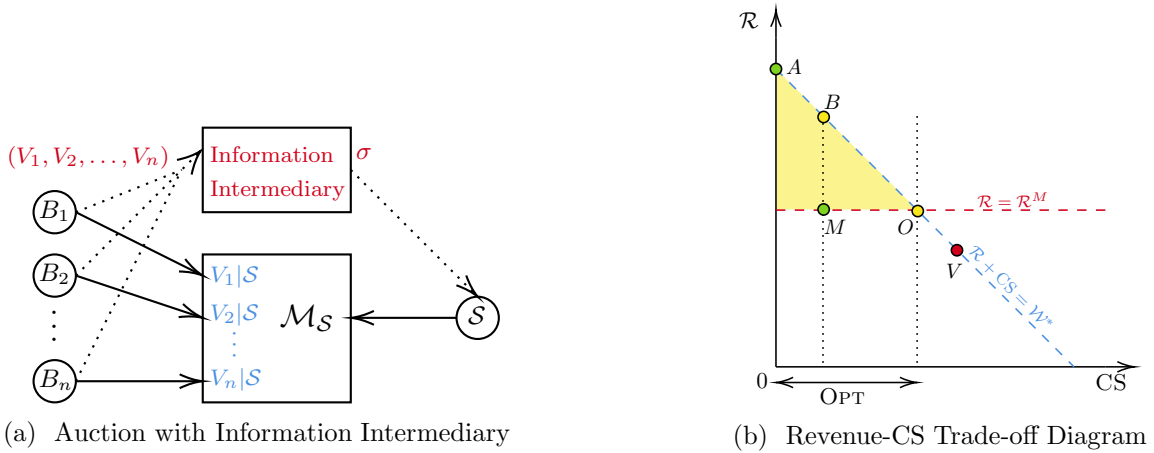
\begin{figure*}[!t]
\begin{minipage}{0.5\columnwidth}
  \centering
  \vspace{1.5cm}
 \resizebox{0.85\columnwidth}{!}{\tikzset{every picture/.style={line width=0.75pt}} 

\begin{tikzpicture}[x=0.75pt,y=0.75pt,yscale=-1,xscale=1]

\draw   (30,50) .. controls (30,44.48) and (34.48,40) .. (40,40) .. controls (45.52,40) and (50,44.48) .. (50,50) .. controls (50,55.52) and (45.52,60) .. (40,60) .. controls (34.48,60) and (30,55.52) .. (30,50) -- cycle ;
\draw   (30,80) .. controls (30,74.48) and (34.48,70) .. (40,70) .. controls (45.52,70) and (50,74.48) .. (50,80) .. controls (50,85.52) and (45.52,90) .. (40,90) .. controls (34.48,90) and (30,85.52) .. (30,80) -- cycle ;
\draw   (30,130) .. controls (30,124.48) and (34.48,120) .. (40,120) .. controls (45.52,120) and (50,124.48) .. (50,130) .. controls (50,135.52) and (45.52,140) .. (40,140) .. controls (34.48,140) and (30,135.52) .. (30,130) -- cycle ;
\draw   (230,100) .. controls (230,94.48) and (234.48,90) .. (240,90) .. controls (245.52,90) and (250,94.48) .. (250,100) .. controls (250,105.52) and (245.52,110) .. (240,110) .. controls (234.48,110) and (230,105.52) .. (230,100) -- cycle ;
\draw  [color={rgb, 255:red, 0; green, 0; blue, 0 }  ,draw opacity=1 ] (110,10) -- (180,10) -- (180,50) -- (110,50) -- cycle ;
\draw  [dash pattern={on 0.84pt off 2.51pt}]  (50,50) -- (100,30) ;
\draw  [dash pattern={on 0.84pt off 2.51pt}]  (50,80) -- (108.59,21.41) ;
\draw [shift={(109,21)}, rotate = 495] [color={rgb, 255:red, 0; green, 0; blue, 0 }  ][line width=0.75]    (10.93,-4.9) .. controls (6.95,-2.3) and (3.31,-0.67) .. (0,0) .. controls (3.31,0.67) and (6.95,2.3) .. (10.93,4.9)   ;
\draw  [dash pattern={on 0.84pt off 2.51pt}]  (50,130) -- (100,30) ;
\draw  [dash pattern={on 0.84pt off 2.51pt}]  (180,30) -- (238.59,88.59) ;
\draw [shift={(240,90)}, rotate = 225] [color={rgb, 255:red, 0; green, 0; blue, 0 }  ][line width=0.75]    (10.93,-3.29) .. controls (6.95,-1.4) and (3.31,-0.3) .. (0,0) .. controls (3.31,0.3) and (6.95,1.4) .. (10.93,3.29)   ;
\draw   (110,70) -- (180,70) -- (180,140) -- (110,140) -- cycle ;
\draw    (50,50) -- (108.21,79.11) ;
\draw [shift={(110,80)}, rotate = 206.57] [color={rgb, 255:red, 0; green, 0; blue, 0 }  ][line width=0.75]    (10.93,-3.29) .. controls (6.95,-1.4) and (3.31,-0.3) .. (0,0) .. controls (3.31,0.3) and (6.95,1.4) .. (10.93,3.29)   ;
\draw    (50,80) -- (108.1,99.37) ;
\draw [shift={(110,100)}, rotate = 198.43] [color={rgb, 255:red, 0; green, 0; blue, 0 }  ][line width=0.75]    (10.93,-3.29) .. controls (6.95,-1.4) and (3.31,-0.3) .. (0,0) .. controls (3.31,0.3) and (6.95,1.4) .. (10.93,3.29)   ;
\draw    (50,130) -- (108,130) ;
\draw [shift={(110,130)}, rotate = 180] [color={rgb, 255:red, 0; green, 0; blue, 0 }  ][line width=0.75]    (10.93,-3.29) .. controls (6.95,-1.4) and (3.31,-0.3) .. (0,0) .. controls (3.31,0.3) and (6.95,1.4) .. (10.93,3.29)   ;
\draw    (230,100) -- (182,100) ;
\draw [shift={(180,100)}, rotate = 360] [color={rgb, 255:red, 0; green, 0; blue, 0 }  ][line width=0.75]    (10.93,-3.29) .. controls (6.95,-1.4) and (3.31,-0.3) .. (0,0) .. controls (3.31,0.3) and (6.95,1.4) .. (10.93,3.29)   ;

\draw (31,43) node [anchor=north west][inner sep=0.75pt]  [font=\footnotesize]  {$B_{1}$};
\draw (31,73) node [anchor=north west][inner sep=0.75pt]  [font=\footnotesize]  {$B_{2}$};
\draw (31,123) node [anchor=north west][inner sep=0.75pt]  [font=\footnotesize]  {$B_{n}$};
\draw (42,94) node [anchor=north west][inner sep=0.75pt]  [rotate=-90] [align=left] {$\displaystyle \dotsc $};
\draw (234,93) node [anchor=north west][inner sep=0.75pt]  [font=\footnotesize]  {$\S$};
\draw (16,17) node [anchor=north west][inner sep=0.75pt]  [font=\scriptsize,color={rgb, 255:red, 208; green, 2; blue, 27 }  ,opacity=1 ]  {$\textcolor[rgb]{0.82,0.01,0.11}{(V_{1} ,V_{2},\ldots,V_{n})}$};
\draw (181,18) node [anchor=north west][inner sep=0.75pt]  [font=\footnotesize,color={rgb, 255:red, 74; green, 144; blue, 226 }  ,opacity=1 ]  {${\textcolor[rgb]{0.82,0.01,0.11}\sigma}$};
\draw (145,92.4) node [anchor=north west][inner sep=0.75pt]  [font=\normalsize]  {$\mathcal{M}_{\S}$};
\draw (112,73.4) node [anchor=north west][inner sep=0.75pt]  [font=\scriptsize,color={rgb, 255:red, 74; green, 144; blue, 226 }  ,opacity=1 ]  {$V_{1} |\S$};
\draw (112,91.4) node [anchor=north west][inner sep=0.75pt]  [font=\scriptsize,color={rgb, 255:red, 74; green, 144; blue, 226 }  ,opacity=1 ]  {$V_{2} |\S$};
\draw (112,122.4) node [anchor=north west][inner sep=0.75pt]  [font=\scriptsize,color={rgb, 255:red, 74; green, 144; blue, 226 }  ,opacity=1 ]  {$V_{n} |\S$};
\draw (127,104.4) node [anchor=north west][inner sep=0.75pt]  [font=\scriptsize,color={rgb, 255:red, 74; green, 144; blue, 226 }  ,opacity=1 ] [rotate=-90] [align=left]  {$\displaystyle \dotsc $};
\draw (112,09) node [anchor=north west][inner sep=0.75pt]  [font=\small] [align=left] {\begin{minipage}[lt]{43pt}\setlength\topsep{0pt}
\begin{center}
\textcolor[rgb]{0.82,0.01,0.11}{{\scriptsize Information}}\\\textcolor[rgb]{0.82,0.01,0.11}{{\scriptsize Intermediary}}
\end{center}
\end{minipage}};

\end{tikzpicture}
  }
 \subcaption{\label{fig1a} Auction with Information Intermediary}
\end{minipage}
\hfill
\begin{minipage}{0.4\columnwidth}
  \centering
  \resizebox{0.85\columnwidth}{!}{\tikzset{every picture/.style={line width=0.75pt}} 

\begin{tikzpicture}[x=0.75pt,y=0.75pt,yscale=-1,xscale=1]

\draw  [dash pattern={on 0.84pt off 2.51pt}]  (137,80) -- (137,240) ;
\draw  [draw opacity=0][fill={rgb, 255:red, 248; green, 231; blue, 28 }  ,fill opacity=0.5 ] (47.5,50.43) -- (137,139.6) -- (47.5,139.6) -- cycle ;
\draw  [dash pattern={on 0.84pt off 2.51pt}]  (77,80) -- (77,240) ;
\draw    (47,240) -- (257,240) ;
\draw    (257,240) -- (273.59,240) ;
\draw [shift={(275.59,240)}, rotate = 540] [color={rgb, 255:red, 0; green, 0; blue, 0 }  ][line width=0.75]    (10.93,-3.29) .. controls (6.95,-1.4) and (3.31,-0.3) .. (0,0) .. controls (3.31,0.3) and (6.95,1.4) .. (10.93,3.29)   ;
\draw    (47,30) -- (47,14.07) ;
\draw [shift={(47,12.07)}, rotate = 450] [color={rgb, 255:red, 0; green, 0; blue, 0 }  ][line width=0.75]    (10.93,-3.29) .. controls (6.95,-1.4) and (3.31,-0.3) .. (0,0) .. controls (3.31,0.3) and (6.95,1.4) .. (10.93,3.29)   ;
\draw    (47,30) -- (47,240) ;
\draw [color={rgb, 255:red, 74; green, 144; blue, 226 }  ,draw opacity=1 ] [dash pattern={on 4.5pt off 4.5pt}]  (47,50) -- (237,240) ;
\draw [color={rgb, 255:red, 208; green, 2; blue, 27 }  ,draw opacity=1 ] [dash pattern={on 4.5pt off 4.5pt}]  (47,140) -- (257,140) ;
\draw  [color={rgb, 255:red, 0; green, 0; blue, 0 }  ,draw opacity=1 ][fill={rgb, 255:red, 126; green, 211; blue, 33 }  ,fill opacity=1 ] (73.8,139.6) .. controls (73.8,137.61) and (75.41,136) .. (77.4,136) .. controls (79.39,136) and (81,137.61) .. (81,139.6) .. controls (81,141.59) and (79.39,143.2) .. (77.4,143.2) .. controls (75.41,143.2) and (73.8,141.59) .. (73.8,139.6) -- cycle ;
\draw  [color={rgb, 255:red, 0; green, 0; blue, 0 }  ,draw opacity=1 ][fill={rgb, 255:red, 208; green, 2; blue, 27 }  ,fill opacity=1 ] (158.8,165.6) .. controls (158.8,163.61) and (160.41,162) .. (162.4,162) .. controls (164.39,162) and (166,163.61) .. (166,165.6) .. controls (166,167.59) and (164.39,169.2) .. (162.4,169.2) .. controls (160.41,169.2) and (158.8,167.59) .. (158.8,165.6) -- cycle ;
\draw  [color={rgb, 255:red, 0; green, 0; blue, 0 }  ,draw opacity=1 ][fill={rgb, 255:red, 126; green, 211; blue, 33 }  ,fill opacity=1 ] (43.4,50) .. controls (43.4,48.01) and (45.01,46.4) .. (47,46.4) .. controls (48.99,46.4) and (50.6,48.01) .. (50.6,50) .. controls (50.6,51.99) and (48.99,53.6) .. (47,53.6) .. controls (45.01,53.6) and (43.4,51.99) .. (43.4,50) -- cycle ;
\draw  [color={rgb, 255:red, 0; green, 0; blue, 0 }  ,draw opacity=1 ][fill={rgb, 255:red, 248; green, 231; blue, 28 }  ,fill opacity=1 ] (133.8,140.6) .. controls (133.8,138.61) and (135.41,137) .. (137.4,137) .. controls (139.39,137) and (141,138.61) .. (141,140.6) .. controls (141,142.59) and (139.39,144.2) .. (137.4,144.2) .. controls (135.41,144.2) and (133.8,142.59) .. (133.8,140.6) -- cycle ;
\draw  [color={rgb, 255:red, 0; green, 0; blue, 0 }  ,draw opacity=1 ][fill={rgb, 255:red, 248; green, 231; blue, 28 }  ,fill opacity=1 ] (73.8,80.6) .. controls (73.8,78.61) and (75.41,77) .. (77.4,77) .. controls (79.39,77) and (81,78.61) .. (81,80.6) .. controls (81,82.59) and (79.39,84.2) .. (77.4,84.2) .. controls (75.41,84.2) and (73.8,82.59) .. (73.8,80.6) -- cycle ;
\draw    (135,250) -- (49,250) ;
\draw [shift={(47,250)}, rotate = 360] [color={rgb, 255:red, 0; green, 0; blue, 0 }  ][line width=0.75]    (10.93,-3.29) .. controls (6.95,-1.4) and (3.31,-0.3) .. (0,0) .. controls (3.31,0.3) and (6.95,1.4) .. (10.93,3.29)   ;
\draw [shift={(137,250)}, rotate = 180] [color={rgb, 255:red, 0; green, 0; blue, 0 }  ][line width=0.75]    (10.93,-3.29) .. controls (6.95,-1.4) and (3.31,-0.3) .. (0,0) .. controls (3.31,0.3) and (6.95,1.4) .. (10.93,3.29)   ;

\draw (22,12.4) node [anchor=north west][inner sep=0.75pt]  [font=\normalsize]  {$\textcolor[rgb]{0,0,0}{\R}$};
\draw (33,242.4) node [anchor=north west][inner sep=0.75pt]    {$0$};
\draw (174.62,159.4) node [anchor=north west][inner sep=0.75pt]  [font=\scriptsize,rotate=-45]  {$\textcolor[rgb]{0.29,0.56,0.89}{\R+\cs=\W^{*}}$};
\draw (249,243.4) node [anchor=north west][inner sep=0.75pt]  [font=\normalsize]  {$\textcolor[rgb]{0,0,0}{\cs}$};
\draw (174,127.4) node [anchor=north west][inner sep=0.75pt]  [font=\scriptsize]  {$\textcolor[rgb]{0.82,0.01,0.11}{\R=\Rm}$};
\draw (76,144.4) node [anchor=north west][inner sep=0.75pt]  [font=\small]  {$M$};
\draw (123,144.4) node [anchor=north west][inner sep=0.75pt]  [font=\small]  {$O$};
\draw (154,172.4) node [anchor=north west][inner sep=0.75pt]  [font=\small]  {$V$};
\draw (52,42.4) node [anchor=north west][inner sep=0.75pt]  [font=\small]  {$A$};
\draw (76,65) node [anchor=north west][inner sep=0.75pt]  [font=\small]  {$B$};
\draw (73,252.4) node [anchor=north west][inner sep=0.75pt]    {$\opt$};

\end{tikzpicture}
  }
  \subcaption{\label{fig1b} Revenue-CS Trade-off Diagram}
\end{minipage}
\caption{\it\small
(a) The auction with information intermediary setting, where the intermediary has full knowledge of valuations $\vec{V}$, and can use this to provide a signal $\sigma$ to the seller and the buyers. The seller then uses the revenue-optimal mechanism $\M_{\S}$ for the posterior distribution over valuations $\S$ given $\sigma$.\\
(b) Two-dimensional space of seller revenue, $\R$, and consumer surplus, $\cs$, of different signaling schemes. The points $M$ and $V$ correspond to Myerson's and VCG mechanisms, and the point $A$ corresponds to selling to the highest-valued buyer at her value when the seller has full information.
}
\label{fig1}
\end{figure*}

\subsection{Our Results}
We consider a single-item auction with $n$ buyers with discrete valuations. We assume the buyer valuations are independent, so $\D = \D_1 \times \D_2 \times \cdots \times \D_n$, where $\D_i$ has support size $\K_i$, and the size of the union of the supports is $\K$.  
We parametrize our results in terms of $n, \K_i$, and $\K$. 

Our first set of results (\cref{sec:lower}) shows a sharp demarcation between the cases of $n=1$ and $n \geq 2$ buyers. In contrast with the former case (where signaling achieves the entire shaded region in~\cref{fig1b}), we show in the latter case, the entire segment $BO$ is not achievable; indeed, the only achievable points on segment $AO$  are arbitrarily close to $A$. Therefore,  achieving full welfare requires sacrificing an arbitrarily large fraction of consumer surplus compared to the no-signaling baseline. 

\begin{theorem} [Proved in \cref{sec:lower}]
\label{thm:lb1}
For any given constant $\varepsilon > 0$, there are instances with $n = 2$ buyers each with $\K_i = 2$, where any signaling scheme $\Z$ under which the revenue-optimal auction obtains full welfare (\latin{i.e.}, allocates to highest-value buyer), has $\cs(\Z) \le \varepsilon \cdot \cs(\D)$, where $\cs(\D)$ is the consumer surplus of Myerson's auction without signaling. 
\end{theorem}

We next ask if we can sacrifice on welfare, but raise a consumer surplus arbitrarily close to $\opt$? We again answer in the negative, and show a lower bound of $2$ on the approximation ratio.
\begin{theorem}[Proved in \cref{sec:lower}]
\label{thm:lb2}
For any constant $\varepsilon > 0$, there are problem instances with $n = 2$ buyers each with $\K_i= 2$,  where any signaling scheme $\Z$ has  $\cs(\Z) \le \left(\frac{1}{2} + \varepsilon\right) \cdot \opt$. 
\end{theorem}
We note that the above results are \emph{existential impossibility results}, and do not depend on the complexity of the signaling scheme.\footnote{Our proofs also imply the same lower bounds when the seller is constrained to run an {\em ex-post} IR mechanism, provided the intermediary's signals induce a product-form posterior distribution.} Overall, the negative results in \cref{thm:lb1,thm:lb2} strongly suggest that in this setting, the focus should be on approximating the consumer surplus. 

The situation improves in \cref{sec:welfare} when we restrict to $\D_i$ that are identical and (discrete-)regular. Here, we first circumvent \cref{thm:lb1} by showing a simple signaling scheme that achieves the point $B$ (i.e., optimal welfare, and same consumer surplus as under Myerson's auction).
One problem that remains, however, is that Myerson's auction may have arbitrarily poor $\cs$: for example, if the $\D_i$ are regular, and chosen such that the reserve price is the highest value in the support, then $\cs=0$, while $\opt > 0$ (and so the approximation factor of Myerson's auction is unbounded). Indeed, even restricting to MHR priors, one can construct instances where the reserve price is close to the maximum value in the support, leading Myerson's auction to have vanishing $\cs$ relative to $\opt$.  This is one reason why getting any non-trivial approximation to $\opt$ is challenging, and we present more discussion in~\cref{sec:intuition}.



In \cref{sec:iid}, we present our main technical result, where we show that when buyers' valuations are drawn from i.i.d. regular distributions, then a simple signaling scheme achieves a constant-approximation to $\opt$.  In more detail, our $\TB_t$ signaling scheme is based on two simple but critical steps: First, the intermediary can use its knowledge of agent valuations to perform a \emph{pre-screening} step that eliminates all but the top-$t$ buyers (for a carefully chosen $t$). Second, given the top $t$ buyers, it can then choose a uniform buyer among this set to serve as a \emph{hold-out buyer}, who the seller can sell to in case she is unable to raise sufficient revenue from the remaining $t-1$ buyers via an auction; this can be achieved by using the single-buyer signaling scheme of~\citet*{bbm2017first} on the chosen buyer. Using a combination of these two ideas, we get the following:

\begin{theorem} [Proved in \cref{sec:iid}]
\label{thm:main0}
When the $\D_i$'s are identical and regular, there is a signaling scheme achieves an $O(1)$-approximation to the optimal consumer surplus $\opt$, and has computation time polynomial in $n$ and $\K$.
\end{theorem}

The nice feature of $\TB_t$ is that it generates signals with posteriors that are (non-identical) product distributions, so that the seller's optimal auction is Myerson's auction~\citep{myerson}, which is also ex-post IC and IR. 
This scheme also turns out to achieve $\opt$ for the special case when $\K = 2$ and $n$ is arbitrary.

Next, we extend this scheme to when the buyers are independent, but not necessarily identical or regular. We obtain the following theorem for this case.

\begin{theorem}
\label{thm:main2}
When the $\D_i$'s are arbitrary, the $\TB_t$ scheme achieves an $O\big(\!\min\!\left(n \log n, \K^2\right)\!\big)$-approximation to  the optimal consumer surplus $\opt$, and has computation time polynomial in $n$ and $\K$.
\end{theorem}


\subsection{Intuition and Techniques}
\label{sec:tech}
For any $n$, the optimal signaling scheme for maximizing surplus can be obtained via an infinite-sized linear program (see~\cref{eq:lp} in~\cref{sec:lower}) with variables for every possible signal, \latin{i.e.}, every possible joint distribution over buyer valuations. Further, for each such signal, the quantity of interest is the consumer surplus of the revenue-optimal auction given the signal. For $n=1$ case, Bergemann \latin{et al.} show this LP has a special structure in that it admits a basis comprising of ``equal-revenue distributions'' containing the revenue-maximizing price (see~\cref{sec:bbm}). Our work shows that this breaks down for optimal auctions with signaling involving $n \geq 2$ buyers.

To understand why things change dramatically from $n=1$ to $n\geq 2$ buyers, in the former case, the optimal mechanism is a simple posted price scheme and its revenue is continuous in the distribution $\D$. However, with multiple buyers, the optimal auction does not have simple structure even for independent buyers (see \cref{alg:my}), and we need to analyze the consumer surplus of this auction, which can be a discontinuous function of the prior. (See \cref{sec:lower} for examples.) Further, for correlated buyers, the revenue of the auction itself may not be continuous in the prior! Indeed, a celebrated result of \citet{CM} shows that slightly perturbing an independent prior to a correlated one can discontinuously increase the revenue to $\W^{*}$, hence decreasing consumer surplus to $0$. (See \cref{CM2} in \cref{sec:prelim}.) This makes it tricky to reason about the optimal signaling scheme, leading to the gap between our upper and lower bounds.

In more detail: Our proofs of \cref{thm:lb1} and \cref{thm:lb2} use a special case of the Cr\'emer-McLean characterization~\citep{CM}: for $n=2$ buyers each with $\K_i = 2$, under any non-independent prior the seller can extract full social surplus as revenue. This lets us focus on signaling schemes where buyers' posterior given each signal are product distributions. Using Myerson's characterization of the optimal auction for discrete valuations~\citep{Elkind}, we show a structural characterization that reduces the space of optimal signals to a sufficiently simple form, yielding the desired counterexamples. Note that we still need to reason about a large space of possible product distributions as signals, which makes our constructions quite non-trivial.

The technically most interesting result in the paper is the $O(1)$-approximation signaling scheme for i.i.d. buyers (\cref{thm:main0} in \cref{sec:iid}). The challenge is the following: Even if we restrict the space of signals so that the posteriors are (non-identical) product distributions, this space is still infinite size, with $\cs$ being a discontinuous function in this space. Our signaling scheme in \cref{sec:iid} balances the trade-off between revealing enough information about valuations so that the item is sold to a high-value buyer, and revealing too much information such that the seller extracts all the surplus.  Balancing these is delicate; nevertheless, our final scheme is simple with polynomial computation time and signal complexity. We present more intuition in \cref{sec:intuition}, where we argue that the guarantee in \cref{thm:main0} cannot be achieved in a straightforward fashion.

\subsection{Related Work}
 The general problem of information structure design considers how sharing additional information can influence the outcome of a mechanism. Different variants of this problem have been formulated and studied; we refer the reader to~\citep{bergemann2019information, dughmi2017algorithmic} for surveys.  Of particular importance to us is the {\em Bayesian persuasion problem} formulated by \citet{kamenica2011bayesian}, where a receiver selects a utility-maximizing action based on incomplete information about the state of nature. A sender who knows the state of nature can signal side-information to the receiver so that the action taken by the receiver is utility-maximizing for the sender. This general problem has been widely studied in different domains such as monopoly pricing and advertising~\citep{bergemann2015limits,chakraborty2014persuasive,xu2015exploring,ParetoIS}. For this problem, there is a distinction between existence and computational results, and the work of~\citet{dughmi2019algorithmic} studies the computational complexity of finding the optimal signaling scheme under different input models.
 
The restriction of our problem to one buyer is the monopoly pricing problem. Here, the intermediary is the sender whose utility is consumer surplus, and the seller is the receiver whose action space is take-it-or-leave-it prices and whose utility is revenue. Beginning with the work of \citet*{bergemann2015limits}, several works~\citep{dughmi2016persuasion,shen2018closed,cummings2020algorithmic,cai2020third,ParetoIS,KoM22,MaoPW22} have considered various extensions and modifications to this basic problem. Unlike monopoly pricing where the buyer is perfectly informed, in our setting, not only the seller, but also all the buyers are receivers, in the sense that they have imperfect knowledge of the true valuations of other buyers, and modify their respective bidding strategies in response to the intermediary's signal to maximize their own utilities. Our setting is therefore a Bayesian persuasion problem with multiple receivers, and this aspect makes it significantly more complex.

There has been work on signaling in auctions that cannot be modeled as Bayesian persuasion, \latin{i.e.}, in which the common signal is not generated by an intermediary who knows all the true values of the buyers. For instance, in the work of \citet{bergemann2007information}, the auctioneer has perfect information about buyer valuations and controls the precision to which buyers can learn it, and in the work of \citet*{hufu2017conceal}, the seller's signal is drawn from a distribution that is correlated with the buyer's value, In both these works, the goal is to maximize seller revenue. Finally, \citet*{shen2019buyer} studies equilibria of optimal auctions when each buyer commits to a signaling scheme with imperfect knowledge of other buyers' valuations, while \citet*{bbm2017first} studies equilibria in first price auctions when buyers are provided correlated signals about other buyers' valuations. In contrast with the former, our work considers a richer space of signals via an information intermediary, while compared to the latter, in our setting the seller's mechanism is not fixed, but is instead also a function of the information structure. 


\section{Preliminaries}
\label{sec:prelim}
We consider Bayesian single-item auctions with $n$ buyers, with independent private valuations $\vec{V}=(V_1,V_2,\ldots,V_n)$ drawn from a known product distribution $\D = \D_1 \times \cdots \times \D_n$. Unless otherwise stated, we present our results for the setting in which each $\D_i$ is discrete.  
We denote by $\K_i$ the size of the support of $\D_i$, and by $\K$ the size of the union of these supports.

For distribution $\D_i$, we use $f_{D_i}$ to denote its probability mass function, and define $S_{D_i}(x)= \Pr_{V_i \sim D_i}[V_i \ge x]$ and $F_{D_i}(x)= \Pr_{V_i \sim D_i}[V_i \le x]$. For a joint distribution $\D$ and vector $\vec{v}$, we use $\Pr[\D = \vec{v}] = f_{\D}(\vec{v})$ as shorthand for denoting the probability of $\vec{v}$ drawn from $\D$. 


\subsection{Revenue-Maximizing Auctions}
\label{sec:optimal}
Given any shared prior $\D'$ on the valuations of the buyers, which in the case of signaling, can be different from $\D$ and arbitrarily correlated, the seller runs an optimal (revenue maximizing) auction that satisfies \latin{ex-post} incentive compatibility and {\em interim} individual rationality. Let $\D'_i$ denote the marginal distribution of buyer $i$. Using the revelation principle~\citep{myerson}, the optimal auction is specified by an allocation rule $x^*(\vec{v}) \ge 0$ and a payment rule $\theta^*(\vec{v})$ (that can be positive or negative) given any realized valuation profile $\vec{v}$. The quantity $x^*_i(\vec{v})$ is the probability buyer $i$ gets the item given the valuation profile $\vec{v}$. The allocation is always non-negative, while the price could be negative. The auction maximizes revenue:
\[
\R(\D') = \max_{\vec{x}, \vec{p}} \sum_{\vec{v}} \Pr[\D' = \vec{v}] \cdot \sum_i \theta_i(\vec{v})
\]
subject to the following constraints. The \emph{ex-post} incentive compatibility (IC) constraint states that the utility of any buyer does not increase by misreporting its valuation. Formally, for every buyer $i$, for every value $q$, and for every valuation vector $\vec{v} = (q,\vec{v}_{-i})$ (where $\vec{v}_{-i}$ denotes the valuations of the other buyers), and every other possible report $r$,
\[
q \cdot x_i(q,\vec{v}_{-i}) - \theta_i(q,\vec{v}_{-i}) \geq q \cdot x_i(r,\vec{v}_{-i}) - \theta_i(r,\vec{v}_{-i}).
\]
The \emph{interim} individual rationality (IR) constraint says that for any buyer $i$ and any value $q$, the expected utility under the mechanism is non-negative:
\[
\sum_{\vec{v}_{-i}} \frac{\Pr[\D' = (q,\vec{v}_{-i})]}{\Pr[\D'_i = q]} \cdot \left(q \cdot x_i(q,\vec{v}_{-i}) - \theta_i(q,\vec{v}_{-i}) \right) \ge 0.
\]

Finally, we have the constraint that the item is allocated probabilistically to at most one buyer:
\[
\sum_i x_i(\vec{v}) \le 1, \qquad \forall \vec{v}.
\]

For any prior $\D'$, let $(\R(\D'),\W(\D'),\cs(\D'))$ denote the expected revenue, welfare (or \emph{total surplus}) and consumer surplus under the revenue-maximizing auction. Then we have $\cs(\D') = \W(\D') - \R(\D')$, and:
\[\R(\D') = \sum_{\vec{v}} \Pr[\D' = \vec{v}] \cdot \sum_i \theta^*_i(\vec{v}) \qquad \mbox{and} \qquad \W(\D') =  \sum_{\vec{v}} \Pr[\D' = \vec{v}] \cdot \sum_i v_i x^*_i(\vec{v}),
\]
where $x^*(\vec{v}) \ge 0$ and $\theta^*(\vec{v})$ are the allocation rule and the payment rule of the optimal auction given any realized valuation profile $\vec{v}$.

Our work builds on two special cases -- independent valuations, and full surplus extraction.

\paragraph{Optimal auction for independent valuations.} When $\D =  \D_1 \times \cdots \times \D_n$ is a product distribution, the optimal auction has a simple form given by \citet{myerson}. 
For distribution $\D_i$ with support $z_1 < z_2 < \cdots < z_k$, its \emph{virtual value} function $\varphi_{\D_i}$ is defined as: 
\begin{equation} 
\label{eq:vv} 
\varphi_{\D_i}(z_k) = z_k \qquad \mbox{and} \qquad \varphi_{\D_i}(z_{\ell}) = z_{\ell} - (z_{\ell+1} - z_{\ell})  \frac{S_{\D_i}(z_{\ell+1})}{f_{\D_i}(z_{\ell})}, \qquad \forall \ell < k.
\end{equation}
If buyer $i$ is the only buyer in the system, the optimal auction sets a fixed price, and the buyer buys the item when her valuation is at least this price. The {\em reserve price} of $\D_i$, denoted $r_{\D_i}$ is the smallest value $r$ in the support of $\D_i$ that maximizes the corresponding revenue $r S_{\D_i}(r)$. It is easy to check that $\varphi_{\D_i}(r_{\D_i}) \ge 0$.

Throughout this paper, we assume the distributions $\D_i$ are {\em regular}, so that $\varphi_{\D_i}(z)$ is a non-decreasing function of $z$. Therefore,  for all $v < r_{\D_i}$, we have $\varphi_{\D_i}(v) < 0$. Our results for the non-i.i.d. case also hold when the distributions are non-regular, by using the non-decreasing ironed virtual value function~\citep{myerson,Elkind} instead. 

For discrete regular distributions, Myerson's auction takes the form~\citep{Elkind} in Algorithm~\ref{alg:my}. 
Note that this auction is also \emph{ex-post} IC and IR.

\begin{algorithm}[htbp]
	\DontPrintSemicolon
Sort the buyers in decreasing order of $q_i = \varphi_{\D_i}(v_i)$. Assume no two values are identical (can be ensured by using a fixed tie-breaking rule). \;
Allocate to the bidder $j$ with highest virtual value $q_j$, provided $q_j\geq 0$. \;
Let $m$ be the bidder with second highest virtual value, and let $w = \max(0, q_{m})$. \;
Charge $j$ the smallest value $z$ in the support of $\D_{j}$ such that $\varphi_{\D_{j}}(z) > w$. \;
\caption{Myerson's Auction with prior $\D$ and valuations $\vec{v}$.}
\label{alg:my}
\end{algorithm}

\paragraph{Extracting full surplus as revenue.} At the other extreme, a celebrated result of \citet{CM} shows that for distributions $\D'$ which are ``sufficiently correlated'', the optimal auction extracts full surplus (\latin{i.e.}, the revenue equals the maximum valuation in each valuation profile). 
Formally, the result requires that for each agent, their conditional distribution over others' values given their own value is full rank; for our purposes, we require a restriction of their result to $n = 2$ buyers, each with two possible valuations.
\begin{theorem}[\citep{CM}]
\label{CM2} 
For $n = 2$ buyers, where each buyer $i$ has $\K_i = 2$ and the joint distribution over the valuations is $\D'$, the seller (who faces an interim IR constraint) can extract the entire social welfare (i.e. get expected revenue equal to the expected value of the maximum of the buyer's valuations) when $\D'$ is a correlated (i.e. not independent) distribution.
\end{theorem}

\subsection{Auctions with an Information Intermediary}
\label{sec:bbm}
We next formalize the model of an information intermediary illustrated in~\cref{fig1a}. Since the effect of the intermediary's signal is captured by the resulting posterior distribution over valuations, for ease of notation, we {henceforth use ``signal'' to refer to a distribution $\S$ over valuations}.

A signaling scheme $\Z = \{\gamma_q, \S_q \}_{q\in[m]}$ comprises a collection of \emph{signals} (\latin{i.e.}, joint distributions over valuations) $\S_1, \S_2, \ldots, \S_m$ and corresponding non-negative weights $\gamma_1, \gamma_2, \ldots, \gamma_m$. The scheme $\Z$ is feasible (or ``Bayes plausible''~\citep{kamenica2011bayesian}) if it satisfies $\sum_q \gamma_q = 1$ and $\sum_q \gamma_q \S_q = \D$.  
The intermediary commits to scheme $\Z$ before the auction, and it is known to the seller and all buyers.

The intermediary maps observed valuation profile $\vec{v} \sim \D$ to signal $\S_q$ with probability $\frac{\gamma_q \Pr[\S_q = \vec{v}]}{\Pr[\D = \vec{v}]}$. The seller uses $\S_q$ as the shared prior and runs an optimal auction on the buyers. 
Note that though $\D$ is a product distribution, the $\{\S_q\}$ can be correlated. Abusing the notations introduced earlier in \cref{sec:optimal}, we denote the revenue generated by signaling scheme $\Z$ as $\R(\Z) =  \sum_q \gamma_q \R(\S_q)$, its consumer surplus by $\cs(\Z) = \sum_q \gamma_q \cs(\S_q)$, and its {welfare} by $\W(\Z)= \sum_q \gamma_q \W(\S_q)$.

When $\D$ is a product distribution, the revenue from any signaling scheme must be at least the optimal revenue of Myerson's auction without signaling, $\R(\D)$. To see this, we note that Myerson's auction on $\D$ is \latin{ex-post} IC and IR. This means that this allocation and payment rule is still a feasible (\latin{interim} IC and IR) mechanism conditioned on receiving any signal, completing the argument. Therefore, the consumer surplus $\cs(\Z)$ under any signaling scheme $\Z$ is bounded by the difference of the maximum possible welfare $\W^* = \E_{\vec{V} \sim \D}[\max_i V_i]$ and the maximum revenue without signaling $\R(\D)$. We henceforth denote this bound as $\opt$, which is defined as follows:
\[
\opt = \W^* - \R(\D).
\]
We say that $\Z$ is a $\tau$-approximation signaling scheme if $\cs(\Z) \ge \frac{\opt}{\tau}$. Our goal is to find the best approximation factor $\tau$ via a signaling scheme whose computation time is polynomial in $n$ and $\K$. In the rest of the paper, we omit the dependence on $\D$ when clear from context.

\paragraph{Optimal signaling for a single buyer.}  
For $n=1$ buyer, \citet{bergemann2015limits} present a signaling scheme with consumer surplus exactly equal to $\opt$ (\latin{i.e.}, implementing the point $O$ in~\cref{fig1b}. Their signaling scheme constructs distributions (signals) $\S_1,\S_2, \ldots, \S_m$ and assigns weights $\gamma_1, \gamma_2, \ldots, \gamma_m$ to them such that $\sum_q \gamma_q \S_q = \D$.   

Let prior $\D$ takes value $v_{i}$ with probability $\eta_{i}$, where $0 < v_1 < \cdots <v_{k}$. Let $\vec{\eta} =(\eta_1, \eta_2, \cdots, \eta_k)$. In each iteration $\ell$, the algorithm constructs an {\em equal revenue} distribution $\S_{\ell}$ and subtracts it from the prior $\D$. This equal revenue distribution assigns positive probability $\eta_{i \ell}$ to $v_i$ if $\eta_i >0$ and assigns $\eta_{i \ell} = 0$ if $\eta_i =0$. In $\S_{\ell}$, the seller raises equal revenue by setting the price to be any of the values $v_i$ with $\eta_i > 0$. It is easy to see that the equal revenue condition specifies a unique distribution $\S_{\ell}$. Note that since this signal is equal revenue, (we may assume) the seller sets the lowest value as price, so that the item always sells and the consumer surplus is maximum possible.

Let $\vec{\eta_{\ell}}$ be the probability vector of $\S_{\ell}$. We set the largest weight $\gamma_{\ell}$ such that $\vec{\eta} - \gamma_{\ell}   \vec{\eta_{\ell}} \ge 0$. We update $\D$ by setting $\vec{\eta}$ to $\vec{\eta} - \gamma_{\ell} \vec{\eta_{\ell}}$,  and increase $\ell$ by one. We repeat this till the support of $\D$ becomes empty. The $\{\gamma_{\ell}, \S_{\ell} \}$ specifies the signaling scheme. We illustrate this procedure by an example.

\begin{example}
Suppose the type space is $\{1,2,3\}$ and $\D = \langle \frac{1}{3}, \frac{1}{3}, \frac{1}{3} \rangle$ are the probabilities of these types. The monopoly price is $\theta = 2$ with revenue $\R(\D) = \frac{4}{3}$, while the point $A$ in~\cref{fig1b} has social welfare $\R(A) = \W^* =  \E[\D] = 2$. Suppose $\S_1 = \langle \frac{1}{2}, \frac{1}{6}, \frac{1}{3} \rangle$ with $\gamma_1 = \frac{2}{3}$; $\S_2 = \langle 0, \frac{1}{3}, \frac{2}{3} \rangle$ with $\gamma_2 = \frac{1}{6}$; and $\S_3 = \langle 0, 1, 0 \rangle$ with $\gamma_3 = \frac{1}{6}$. It is easy to check that the monopoly price for each signal is the lowest price in its support so that the item always sells, and $\sum \gamma_i \R(\S_i) = \frac{4}{3}$. Therefore, $\sum \gamma_i \cs(\S_i) = 2 - \frac{4}{3} = \frac{2}{3} = \opt$, which corresponds to point $O$ in~\cref{fig1b}.
\end{example}

We henceforth use $\BBM(v,D)$ to refer to this scheme when the buyer has valuation distribution $D$ and the realized value is $v \sim D$. Below we state some critical properties of the $\BBM$ scheme which we use in our results.
\begin{lemma}[Implicit in \citep{bergemann2015limits}] 
\label{lem: CS(BBM)} 
For a single buyer with value distribution $\D$ (with reserve price $r_{\D}$), the $\BBM$ mechanism satisfies the following properties:
\begin{enumerate*}
\item For any signal $\S_q$, $\varphi_{\S_q}(v) \ge 0$ for all $v$ in the support of $\S_q$.
\item  $
\cs(\BBM) = \opt
 \ge  \Pr_{V\sim \D} [V < r_{\D}] \cdot \E_{V \sim \D}[V \mid V < r_{\D}]   = \sum_{v < r_{\D}} v f_{\D}(v).$
	\end{enumerate*}
\end{lemma}

\section{Lower Bound Instances}
\label{sec:lower} 
\label{sec:char}
We now prove \cref{thm:lb1,thm:lb2}. We note that though our lower bounds assume that the seller runs interim IR and ex-post IC mechanisms, and the intermediary can send arbitrary signals, the same lower bounds hold even when the seller runs ex-post IC and IR mechanisms, provided we restrict the intermediary's signals to induce posteriors that are product distributions.

Our lower bounds are based on a $2$-buyer instance illustrated in~\cref{fig2}: given values $a > b > c > d$, buyer $1$ has value $V_1 \in \{a,b\}$ with probabilities $\alpha$ and $1-\alpha$ respectively, while buyer $2$ has value $V_2 \in\{c,d\}$ with probabilities $\beta$ and $1-\beta$ respectively. We choose $\alpha a = b$ and $\beta c = d$; thus, the virtual values satisfy: $\varphi_1(a) = a$, $\varphi_2(c) = c$, and $\varphi_1(b) = \varphi_2(d) = 0$. Call this distribution $\D$.

\begin{figure*}[!t]
\centering
\resizebox{0.7\columnwidth}{!}{\input{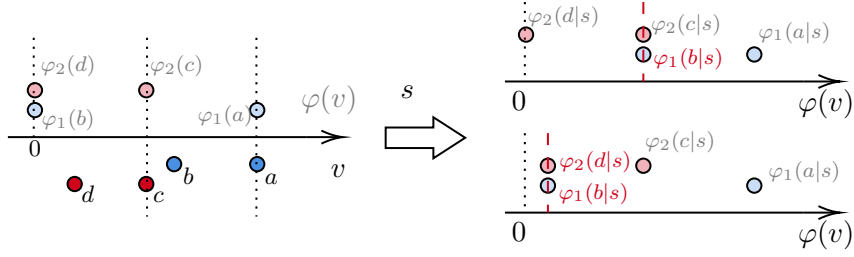}}
\caption{\it\small
Illustrating the setting for~\cref{thm:lb1,thm:lb2}: On the left (below the axis) we show the setting without signaling, where buyer 1 (blue) has values $(a,b)$ and buyer 2 (red) has values $(c,d)$; we also show the corresponding virtual values (above the axis). On the right, we show the two settings characterized by \cref{lem:struct_opt} under which a signal $s$ has non-zero consumer surplus (the changed virtual values are highlighted).
}
\label{fig2}
\end{figure*}

\paragraph{Characterization of optimal signaling.} By~\cref{CM2}, we know any signal that correlates the buyers raises zero consumer surplus. Therefore, the only signals $\S$ of interest are those under which buyer values are independent. Abusing notation we denote such a signal as $s = (\alpha', \beta')$, where $\Pr[v_1 = a] = \alpha'$ and $\Pr[v_2 = c] = \beta'$. 
Note that in this instance, for a signal to get maximum welfare the resulting optimal mechanism must always award buyer $1$, and for non-zero consumer surplus it must award the item to buyer $1$ at price $b$, or buyer $2$ at price $d$.

Let $\cs(s)$ denote the consumer surplus under any such a signal $s$, and $\varphi_1(b|s)$ and $\varphi_2(d|s)$ denote the new virtual values (note that by definition, $\varphi_1(a|s)=a$ and $\varphi_2(c|s)=c$ under any signal $s$ with $\alpha',\beta'>0$). 
We can use Myerson's characterization (\cref{sec:optimal}) to exhaustively characterize the resulting optimal mechanisms as a function of $(\varphi_1(b|s),\varphi_2(d|s))$:
\begin{proposition}
\label{lem:struct_any}
Conditioned on receiving a signal $s$, we have the following cases:
\begin{enumerate*}
	\item If $\varphi_1(b) \geq c$, then the optimal mechanism is to sell to Buyer $1$ at price $b$. $\cs(s) = (a - b)\alpha'$.
	\item If $\varphi_2(d) \geq \max(0, \varphi_1(b))$, then the optimal mechanism is to try selling to Buyer $1$ at price $a$ then to Buyer $2$ at price $d$. $\cs(s) = (1 - \alpha')\beta'(c - d)$.
	\item If $\varphi_1(b) \leq 0$ and $\varphi_2(d) \leq 0$, then the optimal mechanism is to try selling to Buyer $1$ at price $a$ then to Buyer $2$ at price $c$. $\cs(s) = 0$.
	\item If $0 \leq \varphi_1(b) \leq c$ and $\varphi_2(d) \leq \varphi_1(b)$, then the optimal mechanism is to sell to Buyer $1$ at price $b$ if Buyer $2$ has valuation $d$; otherwise, it tries selling to Buyer $1$ at price $a$ then to Buyer $2$ at price $c$. $\cs(s) = \alpha' (1 - \beta') (a - b)$.
\end{enumerate*}
\end{proposition}

Our main insight, however, is that the setting can be further simplified to get the following structural property for the optimal signaling scheme. 

\begin{theorem}[Structural Theorem]
\label{lem:struct_opt}
In an optimal signaling scheme, the only signals $s = (\alpha', \beta')$ that raise non-zero consumer surplus have the following form:
\begin{enumerate}[label=(\arabic*'),nosep]
	\item Under signal $s$, $\varphi_1(b|s) = \varphi_2(c|s) = c$ and $\cs(s) = \alpha'(a - b)$.
	\item Under signal $s$, $\varphi_2(d|s) = \varphi_1(b|s) \geq 0$ and $\cs(s) = \alpha'(1 - \beta')(a - b)$.
\end{enumerate}
\end{theorem}
\begin{proof}
Recall that we restrict ourselves to signals $\S$ under which the buyer valuations remain independent. Any such signal can be alternately written as $s = (\alpha', \beta')$ where $\alpha' = \Pr[v_1 = a]$ and $\beta' = \Pr[v_2 = c]$. For ease of notation, we henceforth drop the conditioning of virtual valuations on signal $s$ (\latin{i.e.}, write $\varphi(\cdot)$ for $\varphi(\cdot|s)$) when clear from context.

Next, let $\gamma_s$ denote the weight of any signal $s = (\alpha', \beta')$. The signaling scheme that maximizes consumer surplus is the solution to the following linear program written over signals $s = (\alpha', \beta')$:
\begin{equation}
\label{eq:lp}
\begin{array}{lccll}
\text{Maximize}  & \displaystyle\sum_s \gamma_s \cs(s) &\\
\text{Subject to}& \sum_{s = (\alpha', \beta')} \gamma_s \alpha' \beta' & \le & \alpha \beta & \\
&\sum_{s = (\alpha', \beta')} \gamma_s \alpha' (1-\beta') & \le & \alpha (1-\beta) & \\
&\sum_{s = (\alpha', \beta')} \gamma_s (1-\alpha') \beta' & \le & (1-\alpha) \beta & \\
&\sum_{s = (\alpha', \beta')} \gamma_s (1-\alpha')(1- \beta') & \le & (1-\alpha)(1- \beta) & \\
&\gamma_s & \ge & 0 & \forall s
\end{array}
\end{equation}

We examine the cases in~\cref{lem:struct_any} with positive consumer surplus, and characterize the optimal solution:
\begin{itemize*}
	\item In Case (1), we have $\varphi_1(b) = c$. To see this, consider any signal $s$ with $\varphi_1(b) > c$. Suppose we increase $\alpha'$ and decrease $\gamma_s$ while preserving the product $\alpha' \gamma_s$. Since $\gamma_s \cs(s) = \gamma_s \alpha'(a-b)$, this is preserved by the change. Therefore, the objective of LP~(\ref{eq:lp}) is preserved, and so are the first two constraints. Further, since $(1-\alpha')$ and $\gamma_s$ decrease, this only makes the third and fourth constraints more feasible. This transformation decreases $\varphi_1(b)$.
	\item In Case (2) and (4), we have $\varphi_2(d) = \varphi_1(b)$.  It does not help to make them unequal by a similar argument as above: In case (2), if $\varphi_2(d) > \varphi_1(b)$, we can increase $\beta'$ while preserving $\gamma_s \beta'$. Since $\gamma_s \cs(s) =  \gamma_s (1 - \alpha')\beta'(c - d)$, this does not change the contribution to the objective of LP~(\ref{eq:lp}), and preserves all constraints. This transformation decreases $\varphi_2(d)$. In case (4), if $\varphi_2(d) < \varphi_1(b)$, we can increase $\alpha'$ while preserving $\gamma_s \alpha'$. Since $\gamma_s \cs(s) =  \gamma_s \alpha' (1 - \beta') (a - b )$,  this does not change the contribution to the objective of LP~(\ref{eq:lp}), and preserves all constraints. This transformation decreases $\varphi_1(b)$.
\end{itemize*}

Therefore, the only two types of signals $s$ that give positive $\cs$ are
\begin{enumerate}[label=(\arabic*')]
	\item If $\varphi_1(b) = c$, then $\cs(s) = \alpha'(a - b)$.
	\item If $\varphi_2(d) = \varphi_1(b) \geq 0$, then $\cs(s) = \max((1 - \alpha')\beta'(c - d), \alpha'(1 - \beta')(a - b))$.
\end{enumerate}
As $(1 - \beta')(b - d) \geq 0$, we have
\[
\left(b - d + \frac{\beta'}{1 - \beta'} (c - d)\right)(1 - \beta') \geq \beta'(c - d).
\]
Notice that in Case (2'), we have $\varphi_1(b) = b - \frac{\alpha'}{1 - \alpha'}(a - b) = d - \frac{\beta'}{1 - \beta'}(c - d) = \varphi_2(d)$.
This gives
\[
\alpha'(1 - \beta')(a - b) \geq (1 - \alpha')\beta'(c - d).
\]
Thus, the two types of signals $s$ that give positive $\cs$ become
\begin{enumerate}[label=(\arabic*')]
	\item If $\varphi_1(b) = c$, then $\cs(s) = \alpha'(a - b)$.
	\item If $\varphi_2(d) = \varphi_1(b) \geq 0$, then $\cs(s) = \alpha'(1 - \beta')(a - b)$. \qedhere
\end{enumerate}
\end{proof}

Using the above structural theorem, the proofs of~\cref{thm:lb1,thm:lb2} follow by different choices of the parameters $(a,b,c,d)$. Suppose the virtual values of $b$ and $d$ are slightly above zero with $\varphi_1(b) > \varphi_2(d)$ so that Case (4) in Proposition~\ref{lem:struct_any} is uniquely optimal for the seller. The optimal auction generates consumer surplus $\cs(\D) = \alpha (1 - \beta) (a - b) =  \frac{b}{c} \cdot \frac{c - d}{a} \cdot (a - b)$ according to \cref{lem:struct_any}.

\subsection{Proof of Theorem~\ref{thm:lb1}}
\label{app:lb1}
To prove \cref{thm:lb1}, we set $b \to c^+$.\footnote{$\varphi_1(b) - \varphi_2(d)$ and $\varphi_2(d)$ can be arbitrarily small as long as positive, so we take the limits for them first, i.e., we are calculating $\lim_{b \to c^+}\lim_{\varphi_2(d) \to 0^+, \varphi_1(b) \to \varphi_2(d)^+} \cs$ in the following part of the proof. This allows us to treat $\alpha = \frac{b}{a}$ and $\beta = \frac{d}{c}$ in calculating $(1 - \alpha)(1 - \beta)$, as $\alpha$ and $\beta$ are not infinitesimally close to $1$ for any fixed $\varepsilon$. (We will define $d = (1 - \varepsilon / 2) b$.)} Now in Proposition~\ref{lem:struct_any}, in Case (1), we must have $\alpha' \rightarrow 0^+$ since $b \rightarrow c^+$, so that $\cs \to 0$. Also if $\alpha' = 1$ in a signal then $\cs = 0$ here. The only other signal where the item is allocated to the higher bidder is in Case (4) when $\beta' = 0$. Let $\gamma$ denote the probability of the signal of this type $s = (\alpha', 0)$. (Having multiple signals of this form gives the same CS as having a single signal as their average.) Since $\varphi_1(b) \geq \varphi_2(d)$, we have $\alpha' \leq \frac{b - d}{a - d}$.

By the constraints of LP~(\ref{eq:lp}), we have:
\[
\Pr[v_1 = b \land v_2 = d] = \gamma (1-\alpha') \le (1-\alpha)(1-\beta),
\]
which simplifies to $\gamma \le \frac{(a-d) \cdot (c-d)}{ac}$. The consumer surplus in this case is therefore:
\[
\cs = \gamma \cs(s) = \gamma \alpha' (a-b) \le \frac{(b-d) \cdot (c-d)}{ac} \cdot (a-b) \le \frac{b-d}{b} \cdot \cs(\D).
\]
Setting $d = \left(1 - \frac{\varepsilon}{2}\right) b$ and combining with the fact that $\cs \to 0$ in Case (1), we have the consumer surplus of any efficient signaling, $\cs \to \frac{\varepsilon}{2} \cdot \cs(\D)$ so $\cs < \varepsilon \cdot \cs(\D)$.


\subsection{Proof of Theorem~\ref{thm:lb2}}
\label{app:lb2}
Without signaling, $\E[\max v_i] = \alpha a  + (1 - \alpha) b$ and $\R(\D) = \alpha a + (1 - \alpha) \beta c$. (Case (2), (3) and (4) in Proposition~\ref{lem:struct_any} give the same revenue $\R(\D)$.) Therefore
\[
\opt = \E[\max v_i] - \R(\D) = (1 - \alpha)(b - \beta c).
\]

Now we assign the values as: $\alpha = \beta = 1 - \delta$; $a = \frac{1}{(1 - \delta)^2}$, $b = \frac{1}{1 - \delta}$, $c = 1$, $d = 1 - \delta$ with $\delta \to 0^+$. Then we plug in the values and the two possible types of signals $s$ in Theorem~\ref{lem:struct_opt} become
\begin{enumerate}[label=(\arabic*')]
	\item If $\alpha' = \frac{1 - \delta}{2 - \delta} < \frac{1}{2}$, then $\cs(s) \le \frac{1}{2} \delta (1 + o(1))$.
	\item If $\alpha' \leq 1 - \delta$; $\beta' = \frac{1 - 3(1 - \alpha') + 3\delta(1 - \alpha') - \delta^2(1 - \alpha')}{1 - 2(1 - \alpha') + \delta(1 - \alpha')} > \frac{1 - 3(1 - \alpha')}{1 - 2(1 - \alpha')}$, then $\cs(s) \le \alpha'(1 - \beta') \delta (1 + o(1))$.
\end{enumerate}

The consumer surplus maximizing signaling scheme should use $t$ signals $S_{2, i}$ of type (2'), with $\alpha'_{2, i}$, $\beta'_{2, i}$ and weight $w(S_{2, i})$. There is an additional signal $S_1$ (with weight $w(S_1)$) of type (1') with $\alpha'_1$ and $\beta'_1$. (Having multiple signals of type (1') gives the same CS as having a single signal as their average.) Denoting the valuation of the first buyer by $v_1$ and the second buyer by $v_2$, the constraints in LP~(\ref{eq:lp}) imply the two constraints:  
\begin{equation}
\tag{Constraint (I)}
\label{eq:1}
\Pr[v_1 = b] = (1-\alpha'_1) w(S_1) + \sum_{i = 1}^t (1 - \alpha_{2, i}')\cdot w(S_{2, i}) \leq 1-\alpha = \delta,
\end{equation}
\begin{equation}
\tag{Constraint (II)}
\label{eq:2}
\Pr[v_1 = b \land v_2 = d]  = \sum_{i = 1}^t (1 - \alpha_{2, i}')(1 - \beta_{2, i}') \cdot w(S_{2, i}) \leq (1-\alpha)(1-\beta) =  \delta^2.
\end{equation}
Note that $\opt = (1 - \alpha)(b - \beta c) =  2\delta^2 (1 + o(1))$. 

The total consumer surplus therefore is:
\begin{align*}
\cs \leq &\frac{1}{2} \delta (1 + o(1)) \cdot w(S_1) + \sum_{i = 1}^t\alpha_{2, i}'(1 - \beta_{2, i}') \delta (1 + o(1)) \cdot w(S_{2, i})\\
\leq &\frac{1}{2} \delta (1 + o(1)) \cdot 2 \left(\delta - \sum_{i = 1}^t (1 - \alpha_{2, i}')\cdot w(S_{2, i})\right) + \delta (1 + o(1)) \sum_{i = 1}^t\alpha_{2, i}'(1 - \beta_{2, i}') \cdot w(S_{2, i})\\
= &\delta^2 (1 + o(1)) + \delta (1 + o(1)) \sum_{i = 1}^t (\alpha_{2, i}'(1 - \beta_{2, i}') - (1 - \alpha_{2, i}')) \cdot w(S_{2, i})\\
\leq &\delta^2 (1 + o(1)) + \delta (1 + o(1)) \sum_{i = 1}^t (1 - \alpha_{2, i}')(1 - \beta_{2, i}') \cdot w(S_{2, i})\\
\leq &\delta^2 (1 + o(1)) + \delta (1 + o(1)) \cdot \delta^2 = \delta^2 (1 + o(1)).
\end{align*}
Here the second inequality follows from \ref{eq:1}, and $\alpha_1' < \frac{1}{2}$ by the condition of (1'). The third inequality uses the implication of $\varphi_2(d) = \varphi_1(b)$ that $\beta_{2, i}' > \frac{1 - 3(1 - \alpha_{2, i}')}{1 - 2(1 - \alpha_{2, i}')}$. The fourth inequality uses \ref{eq:2}. This establishes a lower bound of $2$, since $\opt = 2\delta^2 (1 + o(1))$.


\subsection{Achieving the Pareto-Frontier in the I.I.D. Case} 
\label{sec:welfare}
We now ask if there are cases where we can circumvent \cref{thm:lb1} and maximize welfare while ensuring at least as much surplus as Myerson's auction (that is, achieve a point on the line $BO$ in \cref{fig1b}). Note that \cref{thm:lb1} rules this out for non-i.i.d. distributions. Surprisingly, however, for i.i.d. regular distributions $\D_i$, the following simple signaling scheme turns out to be sufficient for achieving point $B$ in \cref{fig1b}. Morally this shows why our lower bound constructions are delicate.


Suppose the common reserve price of $\D_i$ is $r$, and the maximum value of any buyer is $v_m$. 
\begin{enumerate*}
    \item If $v_m < r$, then the intermediary reveals $v_m$ and the identity of the highest buyer to the seller, who then sells to this bidder at price $v_m$.
    \item If $v_m \ge r$, the intermediary only reveals the information that some buyer has value $\ge r$ (but does not reveal either $v_m$ or the identity of the highest bidder). In this case, though the posterior is not a product distribution anymore, it can be shown that the seller's optimal auction remains the second price auction with reserve $r$.\footnote{Note this is the only case when the seller gets non-zero revenue in the optimal auction for the original product distribution. Suppose for the purpose of contradiction that the seller can do better for this posterior, she can also do better than the optimal auction for the original distribution.}
\end{enumerate*}

It is easy to check that the item always sells to the highest buyer, and the $\cs$ is exactly the same as in Myerson's auction without signaling, thereby achieving point $B$.  Note however that this scheme does not give any guarantees on approximating $\cs$ itself, since the surplus of Myerson's auction is not an approximation to $\opt$. The question of approximating $\opt$ is much more challenging even for the i.i.d. regular setting, and this is what we focus on in the next section.


\section{Approximating Consumer Surplus: The i.i.d. Case}
\label{sec:iid}
In this section, we present our main result (\cref{thm:main0}): An $O(1)$-approximation to $\opt$ when the buyers' valuation distributions $\D_i$ are identical and regular. 

Recall we start with a product distribution $\D = \D_1 \times \D_2 \times \cdots \times \D_n$.  Let $x_1 < x_2 < \cdots < x_{\K}$ be the supports of value distributions. $\R(\D)$ denotes the revenue of the optimal auction (\cref{alg:my}) on $\D$, and $\opt = \E_{\vec{v} \sim \D}[\max_i v_i] - \R(\D)$. 

When $\D_i$'s are identical and regular, since the highest virtual-value and highest value buyers coincide, the optimal auction (\cref{alg:my}) assigns the item to the buyer with highest value if this value is above the common reserve price. Therefore, we decompose $\opt$ into two components:
\begin{itemize}
    \item Myerson's surplus: The $\cs$ generated by Myerson's auction, denoted by $\cs(\D)$.
    \item Non-allocation surplus: The loss in $\cs$ due to Myerson's auction not allocating the item, denoted by $\cs_0$.
\end{itemize}

We therefore have $\opt = \cs(\D) + \cs_0$. In the remainder, we will present an $O(1)$-approximation for the non-allocation surplus $\cs_0$, which will imply \cref{thm:main0} when $\cs_0 \geq \cs(\D)$. (When $\cs(\D) > \cs_0$, sending no signal is already a $2$-approximation.)

\subsection{Preliminaries and Intuition} 
\label{sec:intuition}
Our approximation bound for the non-allocation surplus will also apply when $\D_i$ are independent but not necessarily identical, which will be required for showing \cref{thm:main2}. Therefore, in the sequel, we will proceed assuming the more general case that $\D_i$'s are not necessarily identical, and derive signaling schemes that approximate the non-allocation surplus $\cs_0$ for this case.

We denote a realization from $\D$ by $\vec{v} = \{v_i\}$. Let $p_i = \Pr_{v_i \sim \D_i}[v_i < r_{\D_i}]$ for any buyer $i$, where $r_{\D_i}$ is the reserve price of $\D_i$.  Let $Y_i= \D_{i | < r_{\D_i}}$ denote the distribution of $\D_i$ conditioned on being smaller than $r_{\D_i}$. Suppose we draw a sample independently from each distribution $Y_i$. Let $Z_{\ell}$ denote the distribution for the $\ell^{\text{th}}$ largest value among these $n$ draws. 

We first derive an expression for $\cs_0$. 

\begin{lemma}
\label{lem:z1}
Let $P = \prod _{i=1}^n p_i$. Then, $\cs_0 = P \cdot \E[Z_1]$.
  \end{lemma}
\begin{proof}
Note that $\cs_0$ is the expected surplus lost due to not allocating the item in Myerson's mechanism. This happens only when all realized values are below their corresponding reserve price. In this case, the value lost is the maximum valuation, since this value contributes to the welfare, and the revenue raised is zero. Therefore, we have:
	\[
	\cs_0 = \left(\prod _{i=1}^n p_i \right) \cdot \E\left[\max_{i=1,2,\ldots,n} v_i \ \middle| \ \forall i, v_i < r_{\D_i} \right]
	\]
where the expectation is over $\vec{v} \sim \D$. This is equal to $P \cdot \E[Z_1]$.
\end{proof}

\paragraph{Vanilla Signaling Schemes.} Before presenting our general signaling scheme, it is instructive to first consider simpler schemes to understand the challenges posed by this problem. One possible scheme is to pick a random buyer and apply the single-buyer $\BBM$ signaling scheme defined in Section~\ref{sec:prelim} to it. Denote this buyer by $i$. Such a single-buyer signaling scheme will construct the set of $\BBM$ signals for buyer $i$ by decomposing $\D_i$ and pretending the other buyers don't exist. Given the valuation $v_i \sim \D_i$ of this buyer, the scheme will send a signal $\BBM(v_i,\D_i)$ just as in single-buyer case, and reveal the identity of this buyer. There is no signal sent for the other buyers, so that the seller's information for $j \neq i$ is their prior $\D_j$.

The nice property of the $\BBM$ signaling scheme is that the virtual value of buyer $i$ is always non-negative (\cref{lem: CS(BBM)}). Therefore, in the event when all buyers $j \neq i$ have values $v_j < r_{\D_j}$ (so that their virtual values are negative), the seller allocates the item to $i$. Since $v_i$ is independent of other buyers' values, the mechanism therefore behaves exactly as $\BBM(v_i,\D_i)$ from the perspective of buyer $i$. In other words, with probability $\prod_{j \neq i} p_j$, we generate the single-buyer $\cs$ from \cref{lem: CS(BBM)}, which is at least:
$$\cs(\BBM) \mbox{ from buyer } i \ge \Pr[v_i < r_{\D_i}] \E[v_i \mid v_i < r_{\D_i}] = p_i \E[Y_i].$$ 
Since buyer $i$ is a randomly chosen buyer, the overall $\cs$ generated is:
$$ \mbox{Overall }  \cs \ge \frac{1}{n} \sum_{i=1}^n  \left( \prod_{j \neq i} p_j \right) p_i \E[Y_i] = \left(\prod_{j=1}^n p_j \right) \frac{\sum_{i=1}^n \E[Y_i]}{n} =  P \frac{\sum_{\ell=1}^n \E[Z_{\ell}]}{n} \ge P \frac{\E[Z_1]}{n}$$
where $P = \prod_{j=1}^n p_j$. Therefore, comparing the expression above with that in \cref{lem:z1}, this scheme yields an $n$-approximation to $\cs_0$. Further, we can construct identical regular distributions for which the expected value of the max is comparable to the expected value of the sum, that is, $\sum_{\ell=1}^n \E[Z_{\ell}] =O(\E[Z_1])$. Therefore, this analysis cannot be improved.

On the other hand, the above scheme does achieve $\cs = \opt$ for two-valued i.i.d. distributions ($\K = 2$ and arbitrary $n$). To see this, assume the support is $a < b$, and let $q = \Pr[\D_i = a]$. If the reserve price is $a$, Myerson's auction is already efficient, that is, $\cs(\D) = \opt$; else Myerson's auction has $\cs(\D) = 0$. We now have $Z_i = a$ for all $i$ and $p_i = q$, so that the above scheme has surplus $q^n a = q^n \E[Z_1] = \cs_0 = \opt$. Therefore, in either case, we extract $\cs = \opt$. Interestingly, this also shows that our lower bounds in \cref{thm:lb1,thm:lb2} do require non-i.i.d. distributions when each $\K_i = 2$ regardless of the number $n$ of buyers. 

Moving beyond the $\K = 2$ setting to general $\K$ and $n$, it is tempting to run the $\BBM$ signaling scheme directly on the buyer with highest value, hoping to extract surplus $P \E[Z_1]$. However, this requires revealing the identity of the highest buyer to the seller, since the signaling scheme itself is public knowledge. But if the seller knows the identity of the highest buyer, she can always increase the reserve price to be the second highest bid. In other words, the posterior of the highest buyer is {\em truncated} at $Z_2$. This case needs a more careful construction of the signal and analysis, since the event of a buyer being the largest and hence $\BBM$ being applied to it is now correlated with the surplus this buyer generates in $\BBM$. We perform this construction and analysis in \cref{lem:nonIIDBBMk}. Intuitively, signaling using $\BBM$ on the largest buyer will only yield $\cs \ge P \E[Z_1 - Z_2]$, which is again an $\Omega(n)$-approximation to $P \E[Z_1]$.

Our signaling scheme in the next section  chooses a middle ground between these extremes -- we will choose a rank $t \in \{1,2,\ldots,n\}$ carefully, and choose a buyer whose value lies in the top $t$ ranks at random. We will then perform the single-buyer $\BBM$ scheme on this buyer as we describe below. Surprisingly, this improves the na\"ive $n$-approximation to an $O(1)$-approximation!

\subsection{Ranking-Based Multi-Buyer Signaling Scheme}
\label{sec:tb}
\label{sec:OM}
\label{sec:cstb}
We now introduce the family of signaling schemes $\TB_t$. We will derive a lower bound for the $\cs$ obtained by these schemes in our key lemma,~\cref{lem:nonIIDBBMk}. As mentioned before, since this scheme will also form the basic subroutine for the non-i.i.d. case (\cref{thm:main2}), we present this scheme assuming $\D_i$ can be non-identical. 

Recall the definitions of $p_i, r_{\D_i}, Y_i, Z_{\ell}$ from above. In order to define the signaling scheme, we need an additional definition. For agent $i$ with $V_i\sim\D_i$, we use $\D_{i | >a}$ to denote the conditional distribution of $V_i$ given $V_i>a$, and $\D_{i | <a}$ to denote the conditional distribution of $V_i$ given $V_i<a$. Moreover, we use $\D_{i | >a} - b$ to denote the distribution of $V_i - b$ given $V_i>a$; we refer to it as the distribution of $V_i$ truncated at $a$ and reduced by $b$. 

The following result relates the reserve price of the truncated and the original distributions.
\begin{lemma}
	\label{lem:TRres}  Let $\D'_i =  \D_{i | > \vc} -\vc$ for any $\vc$. Then we have $r_{\D'_i}  \ge r_{\D_i} -\vc $, and moreover, for any $v > \vc$, we have $\varphi_{\D'_i}(v - \vc) = \varphi_{\D_i}(v) - \vc$.
\end{lemma}
\begin{proof}
We first prove the result about reserve prices. Let $r = r_{\D_i}$ and $r' = r_{\D'_i}$. If $r \le v^{\circ}$ the inequality is trivial. Otherwise, suppose for contradiction that for some $r' < r-v^{\circ}$, $r' \cdot S_{\D'_i}(r') \ge (r-v^{\circ}) \cdot S_{\D'_i}(r-v^{\circ}).$ Then we would have:
$(r'+v^{\circ}) \cdot S_{\D'_i}(r') \ge r \cdot S_{\D'_i}(r-v^{\circ}),$ since $S_{\D'_i}(r') \ge S_{\D'_i}(r-v^{\circ})$. Thus, $(r'+v^{\circ}) \cdot S_{\D_i}(r'+v^{\circ}) \ge r \cdot S_{\D_i}(r),$ a contradiction to the assumption that $r$ is the smallest optimal reserve price of $\D_i$.

To see the second part,  if we condition the distribution on $v > v^{\circ}$, this does not change the virtual value function for values $v_j > v^{\circ}$, since both the numerator and denominator in Eq.~(\ref{eq:vv}) scale by the same amount. If we now subtract $v^{\circ}$ from the support of the distribution, it reduces the virtual value by the same amount. This completes the proof.
\end{proof}

Note that $\D'_i =  \D_{i | > \vc} -\vc$ represents buyer $i$'s \emph{excess value} compared to $\vc$. \Cref{lem:TRres} shows that for any threshold $\vc$ and any buyer $i$, given the side-information that $V_i>\vc$, her new reserve price is greater than her original reserve price.

\paragraph{The $\TB_t$ signaling scheme.} We now present the family of signaling schemes $\TB_t$ parameterized by the rank $t \in \{1, \ldots, n\}$. For any realized joint valuation profile $\vec{v}= (v_1,v_2, \ldots, v_n)$, the signal sent by $\TB_t$ consists of two parts. In the first part, $\TB_t$ observes $\vec{v}$ and outputs $(\vc,T)$, where $\vc$ is the value of $(t+1)^{\text{st}}$ largest realized value (or $0$ when $t=n$), and $T$ is the subset of buyers with realized value strictly greater than $\vc$. For the second part of the signal, $\TB_t$ chooses a buyer $j$ uniformly at random from $T$, and computes her excess distribution $\D_{j|>\vc} - \vc$. It then reveals both the identity of $j$, as well as the signal $\BBM(v_j - \vc,\D_{j|>\vc} - \vc)$ generated by the single-buyer $\BBM$ scheme on a buyer with value distribution $\D_{j|>\vc} - \vc$. The scheme is formalized in~\cref{alg:TB}.


\paragraph{Optimal mechanism under $\TB_t$.} Conditioned on receiving the signal generated by $\TB_t$, the seller is guaranteed a revenue of $\vc$ from the $(t+1)^{\text{st}}$ largest buyer, and knows that only buyers in $T$ can pay more than $\vc$. The seller can now charge at least $\vc$ to any buyer in $T$, and can further run an auction over the excess value of buyers in $T$, where for buyer $i \in T$, her excess value has distribution $\D'_i = \D_{i | > \vc} - \vc$. 
Note that for any buyer $i \in T$ except the randomly chosen buyer $j$, a value drawn from $\D'_i$ represents how much more than $\vc$ she is willing to pay. Moreover, distributions $\D'_i$ are independent, and also, since the identity of $j$ is chosen uniformly at random, the $\BBM$ scheme modifies the distribution of buyer $j$ in a fashion that is independent of $\D'_i$.

By~\cref{lem: CS(BBM)}, we know that the $\BBM$ scheme ensures the virtual value of buyer $j$ is always non-negative. 
From the characterization of the optimal auction~\citep{myerson,Elkind}, since the item is always allocated to the highest virtual value buyer as long as this value is non-negative, the item will always be allocated to buyer $j$ if all other buyers $i \in T,i\neq j$ have excess values $v_i - \vc < r_{\D'_i}$ (and hence, negative virtual values).

\begin{algorithm}[H]
	\DontPrintSemicolon
	$ \vc \gets (t+1)^{\text{st}}\text{ largest value in } \vec{v}$\;
	$T \gets \{i : v_i > \vc \}$ \;
	\If{$T \neq \phi$} {
	$j \gets $ Buyer chosen uniformly at random from $T$ \; 
	$s \gets \BBM(v_j-\vc, \D_{j| > \vc}-\vc)$ \;
	\Return $\vc$, $T$, $j$, and $s$ \;
	}
	\Else {
	\Return $\vc$, and $T = \phi$ \;
	}
	\caption{\textsc{$\TB_t (\vec{v}, \D)$}}\label{alg:TB}
\end{algorithm}

\medskip\noindent {\em Consumer surplus under $\TB_t$.} 
We require the following key lemma, that gives a lower bound for the consumer surplus generated under $\TB_t$. This lemma forms the crux of our subsequent analysis, helping us quantify how the $\BBM$ signal recovers much of $\cs$ lost by Myerson's auction. The difficulty in proving it arises because the random choice of buyer $j$ in \cref{alg:TB} is correlated with its rank, which in turn is correlated with its winning the auction and the surplus it generates in $\BBM$. We get around this correlation by carefully {\em coupling} the surplus generated when buyer values are above the reserve with the order statistics of buyer valuations below the reserve.

\begin{lemma}
	\label{lem:nonIIDBBMk}
	For $1 \le t \le n$, the consumer surplus of $\TB_t$ satisfies:
\[\cs(\TB_t) \ge \left(\prod _{i=1}^n p_i \right) \cdot \left(\left(\frac{1}{t} \cdot \sum_{\ell = 1}^t \E[Z_\ell]\right) - \E[Z_{t+1}]\right).\]
\end{lemma} 
\begin{proof}
For convenience, denote $r_i = r_{\D_i}$. Fix a buyer $b$, and any valuation profile $\vec{v}_{-b} = \{v_i,i\neq b\}$ such that $v_i < r_i \,\forall\,i \neq b$.  Define $v^t_b$ as the $t^{\text{th}}$ largest value in $\{v_i,i\neq b\}$. Now consider the event 
\[Q(\vec{v}_{-b}, b, t) = \{ V_b > v^t_b \text{ AND } b \text{ selected for $\BBM$ signaling}\}.\] 
Conditioned on $Q(\vec{v}_{-b}, b, t)$, we have that the $\TB_t$ scheme (\cref{alg:TB}) with parameter $t$ sets threshold value as $\vc = v^t_b$. 
By Lemma~\ref{lem:TRres}, we have that for every $i \in T, i \neq b$, their value $v_i$ is smaller than their new reserve price $\vc + r_{\D'_i}$, since $v_i < r_{\D_i}$, and modifying $\D_i$ to $\D'_i = \D_{i | > \vc} - \vc$ does not decrease the reserve price. 
Therefore, conditioned on $Q(\vec{v}_{-b}, b, t)$, the auction behaves like the single item mechanism $\BBM(v'_b, \D'_b)$, where $v'_b = v_b - v^t_b$, $\D'_b = \D_{b | > v^t_b} - v^t_b$. 
Let $r'_b$ denote the reserve price of $\D'_b$; again using~\cref{lem:TRres} we have $r'_b \ge r_b - v^t_b$. 
Now, using~\cref{lem: CS(BBM)}, we get that the expected consumer surplus generated by $\TB_t$ under $Q(\vec{v}_{-b}, b, t)$ is at least:
\[
\E[\cs(\TB_t) \mid Q(\vec{v}_{-b}, b, t)]  \ge \sum_{v'_b < r'_b} v'_b \Pr[\D'_b = v'_b]  \ge  \sum_{v_b^t < v_b < r_b} (v_b - v^t_b) \frac{f_{\D_b} (v_b)}{S_{\D_b}(v^t_b)}.
\]

Note also that $\Pr[Q(\vec{v}_{-b}, b, t)] = \frac{1}{t}\cdot\left(\prod_{i \neq b} f_{\D_i}(v_i)\mathds{1}_{\{v_i<r_{\D_i}\}} \right) \cdot S_{\D_b}(v^t_b)$. Thus for any buyer $b$, and any valuation profile $\vec{v}_{-b}$ with $v_i < r_i$ for all $i\neq b$, we have
\begin{align*}
\E[\cs(\TB_t)\cdot\mathds{1}_{Q(\vec{v}_{-b}, b, t)}] &= \Pr[Q(\vec{v}_{-b}, b, t)] \cdot \E[\cs(\TB_t) \mid Q(\vec{v}_{-b}, b, t)]\\
&\ge  \frac{1}{t} \left(\prod_{i \neq b} f_{\D_i}(v_i) \right) \sum_{v_b^t < v_b < r_b} (v_b - v^t_b) f_{\D_b} (v_b)\\
&=  \frac{1}{t} \sum_{v_b < r_b}\left(\prod_{i} f_{\D_i}(v_i) \right) \max\{(v_b - v^t_b),0\}.
\end{align*}
For any $\vec{v}$ let $v^{(t)}$ denote the $t^{\text{th}}$ largest value, and $I_t(\vec{v})$ to be the indices corresponding to the top $t$ values in $\vec{v}$.
Summing up over all $b$, and all $\vec{v}_{-b}$ such that $v_i<r_i\,\forall\,i\neq b$, we have
\begin{align}
\label{eq:bcontrib}
\sum_b\sum_{\vec{v}_{-b}} \E[\cs(\TB_t)\cdot\mathds{1}_{Q(\vec{v}_{-b}, b, t)}]
&\geq 
\sum_{\vec{v}|v_i<r_i}\frac{1}{t}\cdot\left(\prod_{i} f_{\D_i}(v_i) \right)\cdot\left(\sum_b \max\{(v_b - v^t_b),0\} \right) \nonumber \\
&=
\sum_{\vec{v}|v_i<r_i}\left(\prod_{i} f_{\D_i}(v_i) \right)\left(\sum_{i\in I_t(\vec{v})}\frac{1}{t}\left(v_i - v^{(t+1)} \right)\right) \nonumber \\
&=
\sum_{\vec{v}|v_i<r_i}\left(\prod_{i} f_{\D_i}(v_i) \right)\left(\left(\sum_{i\in I_t(\vec{v})}\frac{v_i}{t}\right) - v^{(t+1)} \right).
\end{align}

Let $\D_i'' = \D_{i | < r_{\D_i}}$ be the distribution of buyer $i$'s value conditioned on $V_i<r_{\D_i}$. Recall we define $p_i= \Pr_{v_i \sim \D_i}[v_i < r_{\D_i}]$; thus $f_{\D_i''}(v) = f_{\D_i}(v)/p_i$. 
Suppose we independently sample $Y_i\sim\D_{i | < r_{\D_i}}$ for each $i$, and define $Z_{\ell}$ as the $\ell^{\text{th}}$ largest value in $\{Y_i\}$.
Then~\cref{eq:bcontrib} can be written as
\begin{align*}
\sum_b\sum_{\vec{v}_{-b}|v_i<r_i\,\forall\,i\neq b} \E[\cs(\TB_t)\cdot\mathds{1}_{Q(\vec{v}_{-b}, b, t)}]
&\geq 
\sum_{\vec{v}|v_i<r_i}\left(\prod_{i}p_i  \right)f_{\D''}(\vec{v})\left(\left(\sum_{i\in I_t(\vec{v})}\frac{v_i}{t}\right) - v_{(t+1)} \right)\\
&=
\left(\prod_{i}p_i\right)\left(\left(\sum_{\ell=1}^t\frac{\E[Z_{\ell}]}{t}\right) - \E[Z_{t+1}] \right).
\end{align*}
Finally, noting that the $Q(\vec{v}_{-b}, b, t)$ events are all non-overlapping, we can write
\begin{align*}
\cs(\TB_t) &\ge \sum_b\sum_{\vec{v}_{-b}|v_i<r_i\,\forall\,i\neq b} \E[\cs(\TB_t)\cdot\mathds{1}_{Q(\vec{v}_{-b}, b, t)}],
\end{align*}  
thereby completing the proof.
\end{proof}

\subsection{Approximating Non-allocation Surplus $\cs_0$ and Proof of Theorem~\ref{thm:main0}} 
\label{app:cs0}
Given the above signaling scheme, approximating $\cs_0$ (and hence showing \cref{thm:main0}) is now simple: We choose the parameter $t \in \{1,2,\ldots,n\}$ that maximizes $\cs(\TB_t)$ and run $\TB_t$.  We denote this scheme as $\{\S_0 (\vec{v},\D)\}_{\vec{v}}$, and present it in~\cref{alg:CS0_2}.

\begin{algorithm}[H]
	\DontPrintSemicolon
	Choose $t = \mbox{argmax}_{t'=1}^n \cs(\TB_{t'})$. 
	\Return $\TB_t (\vec{v},\D)$ \;
	\caption{\textsc{$\S_0 (\vec{v},\D)$}}\label{alg:CS0_2}
\end{algorithm}

The following theorem shows that this scheme approximates $\cs_0$ when the $\D_i$ are a common regular distribution $\Theta$. Since $\opt = \cs(\D) + \cs_0$ when $\D_i$ are regular and identical, \cref{thm:cs0} shows the better of no signaling and \cref{alg:CS0_2} is an $O(1)$-approximation to $\opt$, completing the proof of~\cref{thm:main0}. 

\begin{theorem}
\label{thm:cs0}
The consumer surplus of the signaling scheme $\S_0 (\vec{v},\D)$  is an $O(1)$-approximation to the non-allocation surplus, $\cs_0$.
\end{theorem}




Let $\alg$ denote the $\cs$ of $\S_0 (\vec{v},\D)$. The above theorem will follow from the following, since $\alg$ is at least  the LHS by \cref{lem:nonIIDBBMk} and $\cs_0$ is equal to $P \cdot \E[Z_1]$ by \cref{lem:z1}.

\begin{theorem}
\label{thm:regularity}
When each $\D_i$ is i.i.d. and regular with common distribution $\Theta$, then we have:
$$ P \cdot \max_t \left(\left(\frac{1}{t} \sum_{i = 1}^t \E[Z_i]\right) - \E[Z_{t + 1}]\right)  \geq \frac{1}{1900} \cdot P \cdot \E[Z_1].$$
\end{theorem}

The rest of this section is devoted to proving \cref{thm:regularity}.  Note that the above theorem is not true if the distribution $Y$ (the conditional distribution of $\Theta$ below its reserve) on which the order statistics $Z_i$ are defined, is a generic regular distribution such as an Exponential distribution. Consider the following example:

\begin{example}
Suppose $Y$ is \texttt{Exponential}$(1)$. Then $\E[Z_i] = \sum_{j=i}^n \frac{1}{j} = H_n - H_{i-1}$, where $H_i = \sum_{j=1}^i \frac{1}{j}$. Therefore $\cs_0 = \E[Z_i] = H_n$. It is easy to check that $\left(\frac{1}{t} \sum_{i = 1}^t \E[Z_i]\right) - \E[Z_{t + 1}] = 1 + H_n - H_{t} - \left(H_n - H_{t}\right) = 1.$
Therefore, $\cs_0 = \Omega(\log n) \cdot \alg$.
\end{example}

Note however that $Y$ is of a more specific form: It is the conditional distribution of a regular distribution $\Theta$ below its reserve $r$. This means in particular that $Y$ \emph{cannot} be an Exponential distribution (or its discrete counterpart, the Geometric distribution). We will crucially use the property that the revenue of $\Theta$ when the price is set below the reserve $r$ is a concave function of the quantile of the price, and further, this function is non-decreasing. Formally, we use the following.

\begin{lemma}
For any regular distribution $\Theta$ with $S(v) = \Pr[\Theta \ge v]$ and reserve $r$, and for any $u \leq v < r$ in the support of $\Theta$, we have: (1) $u \cdot S(u) \geq \frac{1 - S(u)}{1 - S(v)} \cdot v \cdot S(v),$ and (2) $u \cdot S(u) \le v \cdot S(v).$ 
\label{lem:reg_concave}
\end{lemma}
\begin{proof}
Let $f_i = \Pr[\D = i]$. Regularity means the quantity $\varphi(x_i) = x_i - (x_{i+1} - x_i) \frac{S(x_{i+1})}{f(x_{i})}$ is non-decreasing in $i$. Let $R_i = x_i S(x_i)$ be the revenue when the price is $x_i$. Note that 
\[
\varphi(x_i) = \frac{R_{i} - R_{i+1}}{S(x_{i}) - S(x_{i+1})}.
\]
Consider the revenue as a function of $S(x_i)$. Regularity means this function is concave on $S(x_i)$ where $x_i$ is within the support of $\D$. The lemma then follows since the curve is non-increasing in $S(x_i)$ for $x_i \le r$.
\end{proof}

The next idea in the analysis is to decompose the distribution $Y$ into the ``core'' and ``tail'', so that values $v$ with $\Pr[Y \ge v] = O(1/n)$ lie in the tail, and the rest lie in the core. Roughly speaking, we will show that the expected value of $Y$ in either the core or the tail is upper bounded by $\cs(\rank_t)$ for suitable choices of $t$. Such choices are non-trivial and form the crux of our analysis. 

\subsubsection{Core-Tail Decomposition} 
To fix notation, recall that every $\D_i$ is the same regular distribution $\Theta$. The values $x_1 < x_2 < \cdots < x_\K$ form the support of  $\Theta$ and its reserve price is $r$. 
Recall $S(x) = \Pr[\Theta \ge x]$ and $Y$ is the conditional distribution of $\Theta$ strictly below $r$.

First, if $n \leq 1000$, $\alg \geq \frac{1}{1000} \E[Z_1] = \frac{1}{1000} \cs_0$. Therefore, we assume $n > 1000$ in the following analysis. We further assume $r > x_1$ since otherwise $\cs_0 = 0$ and Myerson's auction itself raises optimal $\cs$. Suppose $r = x_i$ with $i > 1$, then  $\{x_1, \ldots, x_{i-1}\}$ is the support of $Y$.

Let $u^+$ be the smallest value in the support of $Y$ with $\Pr[Y \geq u^+] \leq \frac{4}{n}$ and let $u^-$ be the largest value in the support of $Y$ with $\Pr[Y \geq u^-] > \frac{4}{n}$. Both of these values exist since $\Pr[Y \ge x_1] = 1$ (we assumed $r > x_1$) and $\Pr[Y \ge r] = 0$. Further, $u^-$ and $u^+$ are consecutive values in the support of $Y$.

We divide $\cs_0$ into $\core$ and $\tail$, based on whether there is a realized value at least $u^+$. 

\begin{definition}
$\core := P \cdot u^-$, and $\tail := P \cdot n \cdot \E[Y \cdot \ind(Y \geq u^+)]$.
\end{definition}

\begin{lemma}
\label{lem:coretail}
$\cs_0 \leq \core + \tail$.
\end{lemma}
\begin{proof}
This comes immediately from
\begin{align*}
\cs_0 &= P \cdot (\E[Z_1 \cdot \ind(Z_1 \leq u^-)] + \E[Z_1 \cdot \ind(Z_1 \geq u^+)])\\
&\leq P \cdot (u^- \cdot \Pr[Z_1 \leq u^-] + n \cdot \E[Y \cdot \ind(Y \geq u^+)])\\
&\leq \core + \tail. 
\end{align*}
where the first inequality upper bounds the max of $n$ i.i.d. samples from $\ind(Y \geq u^+)$ by their sum.
\end{proof}

\subsubsection{Upper Bound on $\core$ and $\tail$}
We will separately bound $\core$ and $\tail - O(1) \cdot \core$ in terms of $\alg$. We first bound $\tail$ using the following lemma.

\begin{lemma}
\label{lem:tail}
$\alg \geq \frac{1}{250} \cdot (\tail - 4 \cdot \core)$.
\end{lemma}
\begin{proof}
If $\Pr[Y \geq u^+] = 0$, then the inequality trivially holds since $\tail = 0$. Otherwise,
\begin{align*}
\alg \geq \cs(\rank_1) & = P \cdot \E[Z_1 - Z_2] \\
& \geq P \cdot \Pr[Z_1 \geq u^+ \land Z_2 \leq u^-] \cdot \E[Z_1 - Z_2 \mid Z_1 \geq u^+ \land Z_2 \leq u^-]\\
&\geq P \cdot \Pr[Z_1 \geq u^+ \land Z_2 \leq u^-] \cdot (\E[Z_1 \mid Z_1 \geq u^+ \land Z_2 \leq u^-] - u^-).
\end{align*}
Now, to bound the term $\Pr[Z_1 \geq u^+ \land Z_2 \leq u^-]$, we have the following. Note here $u^-$ and $u^+$ are consecutive values in the support of $Y$.
\begin{align*}
\Pr[Z_1 \geq u^+ \land Z_2 \leq u^-] &= \Pr[Z_1 \geq u^+] \cdot  \frac{\Pr[Z_1 \geq u^+ \land Z_2 \leq u^-]}{\Pr[Z_1 \geq u^+]}
\\
& =\Pr[Z_1 \geq u^+] \cdot \frac{n \cdot (1 - \Pr[Y \leq u^-]) \cdot \Pr[Y \leq u^-]^{n - 1}}{1 - \Pr[Y \leq u^-]^{n}}\\
& =\Pr[Z_1 \geq u^+] \cdot \frac{n}{\sum_{i=0}^{n-1} \Pr[Y \leq u^-]^i} \cdot \Pr[Y \leq u^-]^{n - 1}
\\
&\geq\Pr[Z_1 \geq u^+] \cdot \Pr[Y \leq u^-]^{n - 1} \\
& \geq \Pr[Y \leq u^-]^{n} = \Pr[Z_1 \geq u^+] \cdot \Pr[Z_1 \leq u^-].
\end{align*}
We now bound $ \Pr[Z_1 \geq u^+]$ and $\Pr[Z_1 \leq u^-]$ separately. To bound $ \Pr[Z_1 \geq u^+]$, note that 
\begin{align*}
\Pr[Z_1 \geq u^+] & = 1 - (1 - \Pr[Y \ge u^+])^n 
 \geq 1 - \exp(-n \cdot \Pr[Y \ge u^+])  \geq \frac{1 - e^{-4}}{4} \cdot n \cdot \Pr[Y \ge u^+]
\end{align*}
where the final inequality follows since the function $\psi(x) = \frac{1 - e^{-4x}}{(1 - e^{-4}) x} \geq 1$ when $x = \frac{1}{4} \cdot n \cdot \Pr[Y \ge u^+] \leq 1$. To bound $\Pr[Z_1 \leq u^-]$, note that:
$$
\Pr[Z_1 \leq u^-]  = (1 - \Pr[Y \geq u^+])^n  \geq (1 - 4/n)^n \geq 0.018
$$
where we have used $n \geq 1000$ in the last inequality. Therefore, we have
$$ \Pr[Z_1 \geq u^+ \land Z_2 \leq u^-]  \geq \Pr[Z_1 \geq u^+] \cdot \Pr[Z_1 \leq u^-] \geq 0.004 \cdot n \cdot \Pr[Y \geq u^+].$$

To bound $\alg$, we also need to simplify the term $\E[Z_1 \mid Z_1 \geq u^+ \land Z_2 \leq u^-]$. For this, let $Y_1, \ldots, Y_n$ be independent draws from $Y$. We have:
\begin{align*}
\E[Z_1 \mid Z_1 \geq u^+ \land Z_2 \leq u^-] &= \sum_{i = 1}^n \frac{1}{n} \cdot \E[Y_i \mid Y_i \geq u^+ \land Y_j \leq u^-, \ \forall j \neq i]\\
&= \E[Y_1 \mid Y_1 \geq u^+ \land Y_j \leq u^-, \ \forall j \neq 1]\\
&= \E[Y_1 \mid Y_1 \geq u^+]= \E[Y \mid Y \geq u^+].
\end{align*}
Here, the first equality follows since any $Y_i$ is equally likely to be the maximum value, and the third equality follows by the independence of $Y_i$'s. Putting all this together, we bound $\alg$ as:
\begin{align*}
\alg & \geq P \cdot \Pr[Z_1 \geq u^+ \land Z_2 \leq u^-] \cdot (\E[Z_1 \mid Z_1 \geq u^+ \land Z_2 \leq u^-] - u^-)
\\ & \geq 0.004 P \cdot n \cdot \Pr[Y \geq u^+] \cdot (\E[Y \mid Y \geq u^+] - u^-)\\
& = 0.004 P \cdot ( n \cdot \E[Y \cdot \ind(Y \geq u^+)] - n \cdot \Pr[Y \geq u^+] \cdot u^-) \\
&\geq 0.004 \cdot (\tail - 4 \cdot \core). 
\end{align*}
where the last inequality uses $n \cdot \Pr[Y \geq u^+] \le 4$.
\end{proof}

Next, we bound $\core$, crucially using the regularity of $\Theta$ (via \cref{lem:reg_concave}).

\begin{lemma}
\label{lem:core}
$\alg \geq \frac{1}{330} \cdot \core$.
\end{lemma}
\begin{proof}
We divide the proof into four cases based on the value of $S(u^-)$.

\paragraph{Case (1): $S(u^-) \geq 0.1$ and there is a value $w$ with $S(w) \in [0.75 \cdot S(u^-) + 0.25, 0.25 \cdot S(u^-) + 0.75]$.} We have $w < u^- < r$. Note that $S(u^-) \geq 0.1$ and $S(w) \le 0.25 \cdot S(u^-) + 0.75$ together imply $S(u^-) \ge \frac{0.1}{0.775} \cdot S(w)$. Also note that $S(w) \geq 0.75 \cdot S(u^-) + 0.25$ implies $1 - S(w) \ge \frac{1}{4} (1 - S(u^-)) $. Using \cref{lem:reg_concave}, we now have
\[
w \cdot S(w) \geq \frac{1 - S(w)}{1 - S(u^-)} \cdot u^-  \geq \frac{1}{4} \cdot u^- \cdot \frac{0.1}{0.775} \cdot S(w).
\]
so that $w \geq \frac{1}{31} \cdot u^-$. To bound the performance of $\alg$, we have
\begin{align*}
\alg &\geq \cs(\rank_n) = P \cdot \frac{1}{n} \sum_{i=1}^n \E[Z_i]
\ge P \cdot \frac{\lfloor n/10 \rfloor}{n} \cdot \E[Z_{\lfloor n/10 \rfloor}] \\
& \geq P  \cdot \frac{\lfloor n/10 \rfloor}{n} \cdot w \cdot \Pr[Z_{\lfloor n/10 \rfloor} \geq w]\\
&\geq P \cdot (1 - \Pr[Z_{\lfloor n/10 \rfloor} < w]) \cdot \frac{1}{320} \cdot u^-\\
&= \frac{1}{320} \cdot (1 - \Pr[Z_{\lfloor n/10 \rfloor} < w]) \cdot \core.
\end{align*}
where we have used $w \geq \frac{1}{31} \cdot u^-$ and $n > 1000$.

Now notice that $\Pr[Y \geq w] \geq \Pr[\Theta \geq w] = S(w) \ge 0.25$. Consider the sum of $n$ i.i.d. samples from Bernoulli$(1, 0.25)$.  The quantity $\Pr[Z_{\lfloor n/10 \rfloor} < w]$ is upper bounded by the probability that the above sum is at most $\lfloor n/10 \rfloor$. We apply Chernoff bounds~\citep{Mitzenmacher} to this Bernoulli sum to obtain:
\[
\Pr[Z_{\lfloor n/10 \rfloor} < w] \leq \exp\left(-\frac{(0.25n - 0.1n)^2}{2 \cdot (0.25n)}\right) \leq \exp\left(-0.045 \cdot n\right) < 0.001.
\]
where we have used $n > 1000$.

We conclude that $\alg \geq \frac{1}{330} \cdot \core$ in this case.

\paragraph{Case (2): $S(u^-) < 0.1$ and there is a value $w$ with $S(w) \in [2 \cdot S(u^-), 4 \cdot S(u^-)]$.} As before, we have $w < u^- < r$. Further, note that $S(u^-) < 0.1$ and $S(w) \le 4 \cdot S(u^-)$ together imply $1 - S(w) \ge \frac{0.6}{0.9}(1 - S(u^-))$ and $S(u^-) \ge \frac{1}{4} S(w)$. Using \cref{lem:reg_concave}, we have
\[
w \cdot S(w) \geq \frac{1 - S(w)}{1 - S(u^-)} \cdot u^- \cdot S(u^-) \geq \frac{0.6}{0.9} \cdot u^- \cdot \frac{1}{4} \cdot S(w).
\]
so that $w \geq \frac{1}{6} \cdot u^-$. 

Let $b$ be the largest value within the support of $Y$ with $\Pr[Y \geq b] \geq 2 \cdot \Pr[Y \geq w]$. This implies $\Pr[Y > b] < 2 \Pr[Y \geq w]$. Note that such a value $b$ exists since $\Pr[Y \geq x_1] = 1$. Since $b < w < r$, by \cref{lem:reg_concave}, we have $b \cdot S(b) \leq w \cdot S(w)$ 
and thus
\[
b \leq w \cdot \frac{S(w)}{S(b)} = w \cdot \frac{\left(S(w) - S(r)\right) + S(r)}{\left(S(b) - S(r)\right) + S(r)}.
\]
Further, using the fact that $\frac{S(w) - S(r)}{S(b) - S(r)} = \frac{\Pr[Y \geq w]}{\Pr[Y \geq b]} \leq \frac{1}{2}$, and $S(w) - S(r) \geq S(r)$ (this is since $S(w) \geq 2 \cdot S(u^-)$), we get $b \leq \frac{2}{3} \cdot w$ and thus
\[
w - b \geq \frac{1}{3} \cdot w \geq \frac{1}{18} \cdot u^-.
\]

Choose $t \geq 40$ so that $\Pr[Y \geq w] \in [0.2 \cdot t/n, 0.3 \cdot t/n]$. We now show that such a $t$ exists. First note that these intervals overlap for consecutive $t$. Next note that $\Pr[Y \geq w] \geq 2 \cdot \Pr[Y \geq u^-] \geq \frac{8}{n}$ (by the choice of $u^-$), so that when $t=40$, the value $0.2 \cdot t/n$ is a lower bound on $\Pr[Y \geq w]$. Further, since $\Pr[Y \geq w] \leq 4 \Pr[Y \geq u^-] \le 0.4$, this means that when $t=n$, the value $0.3 \cdot t/n$ is an upper bound on $\Pr[Y \geq w]$. This shows such a $t$ must exist.  

Putting these together, we now have:
\begin{align*}
\alg \geq \cs(\rank_t) & = P \cdot \frac{1}{t} \sum_{i=1}^t \E[Z_i - Z_{t+1}] 
 \geq P \cdot \frac{\lfloor t/10 \rfloor}{t} \cdot \E[Z_{\lfloor t/10 \rfloor} - Z_{t+1}]
\\
& \geq P \cdot \frac{\lfloor t/10 \rfloor}{t} \cdot \Pr[Z_{\lfloor t/10 \rfloor} \geq w \land Z_{t} \leq b] \cdot  (w - b)\\
&\geq P \cdot (1 - \Pr[Z_{\lfloor t/10 \rfloor} < w] - \Pr[Z_{t} > b]) \cdot \frac{1}{190} \cdot u^-\\
&= \frac{1}{190} \cdot (1 - \Pr[Z_{\lfloor t/10 \rfloor} < w] - \Pr[Z_{t} > b]) \cdot \core.
\end{align*}
where we have used $w - b \geq \frac{1}{18} \cdot u^-$ and $n > 1000$. Now notice that $\Pr[Y \geq w] \geq 0.2 \cdot t/n$ and $\Pr[Y > b]  < 2S_Y(w) \le 0.6 \cdot t/n$. We apply Chernoff bounds as in Case (1) to obtain:
\[
\Pr[Z_{\lfloor t/10 \rfloor} < w] \leq \exp\left(-\frac{(0.2 t - 0.1 t)^2}{2 \cdot (0.2 t)}\right) \leq \exp\left(-0.025 \cdot t\right) < 0.368.
\]
and
\[
\Pr[Z_{t} > b] \leq \exp\left(-\frac{(t - 0.6 t)^2}{(1 + 1/0.6) \cdot (0.6 t)}\right) \leq \exp\left(-0.1 \cdot t\right) < 0.019.
\]
where we have used $t \ge 40$. Combining this with the lower bound on $\alg$, we derive:
\[
\alg \geq \frac{1}{190} \cdot (1 - \Pr[Z_{\lfloor t/10 \rfloor} < w] - \Pr[Z_{t} > b]) \cdot \core \geq \frac{1}{310} \cdot \core.
\]

\paragraph{Case (3): $S(u^-) \geq 0.1$ and there is no value $w$ with $S(w) \in [0.75 \cdot S(u^-) + 0.25, 0.25 \cdot S(u^-) + 0.75]$.}
Let $w^+$ be the smallest value within the support of $\Theta$ with $S(w^+) < 0.75 \cdot S(u^-) + 0.25$, and let $w^-$ be the largest value within the support of $\Theta$ with $S(w^-) > 0.25 \cdot S(u^-) + 0.75$. Using the same argument as in the previous cases, notice that $w^+ \geq \frac{1}{5} \cdot u^-$, since
\[
w^+ \cdot S(w^+) \geq \frac{1 - S(w^+)}{1 - S(u^-)} \cdot u^- \cdot S(u^-) \geq \frac{3}{4} \cdot u^- \cdot \frac{0.1}{0.325} \cdot S(w^+).
\]
Furthermore, $w^-$ and $w^+$ are consecutive in the support of $Y$ and both values are less than $r$.

If $\Pr[Y \geq w^+] > 0.25$, the proof is similar to Case (1) that
\begin{align*}
\alg &\geq \cs(\rank_n) \geq P \cdot \Pr[Z_{\lfloor n/10 \rfloor} \geq w^+] \cdot \frac{\lfloor n/10 \rfloor}{n} \cdot w^+\\
&\geq P \cdot (1 - \Pr[Z_{\lfloor n/10 \rfloor} < w^+]) \cdot \frac{1}{60} \cdot u^-\\
&\geq \frac{1}{60} \cdot (1 - \Pr[Z_{\lfloor n/10 \rfloor} < w^+]) \cdot \core,
\end{align*}
and by Chernoff bounds analogous to Case (1) we have,
\[
\Pr[Z_{\lfloor n/10 \rfloor} < w^+] \leq \exp\left(-\frac{(0.25n - 0.1n)^2}{2 \cdot (0.25n)}\right) \leq \exp\left(-0.045 \cdot n\right) < 0.1.
\]

Now assume $\Pr[Y \geq w^+] \leq 0.25$. Choose $t \geq 20$ so that $\Pr[Y \geq w^+] \in [0.2 \cdot t/n, 0.25 \cdot t/n]$. Using the same argument as for Case (2), this $t$ exists since $\Pr[Y \geq w^+] \geq \Pr[Y \geq u^-] \geq \frac{4}{n}$ and $\Pr[Y \geq w^+] \leq 0.25$. Similar to Case (2), we have
\begin{align*}
\alg &\geq P \cdot \Pr[Z_{\lfloor t/10 \rfloor} \geq w^+ \land Z_t \leq w^-] \cdot \frac{\lfloor t/10 \rfloor}{t} \cdot (w^+ - w^-)\\
&\geq P \cdot (1 - \Pr[Z_{\lfloor t/10 \rfloor} < w^+] - \Pr[Z_t > w^-]) \cdot \frac{\lfloor t/10 \rfloor}{t} \cdot (w^+ - w^-)\\
&= P \cdot (1 - \Pr[Z_{\lfloor t/10 \rfloor} < w^+] - \Pr[Z_t \geq w^+]) \cdot \frac{\lfloor t/10 \rfloor}{t} \cdot (w^+ - w^-).
\end{align*}
where we have used that $w^+$ and $w^-$ are consecutive values in the support of $Y$.

Now notice that $\Pr[Y \geq w^+] \geq 0.2 \cdot t/n$ and $\Pr[Y \geq w^+] \leq 0.25 \cdot t/n$ and as in Case (1), we apply Chernoff bounds as follows:
\[
\Pr[Z_{\lfloor t/10 \rfloor} < w^+] \leq \exp\left(-\frac{(0.2 t - 0.1 t)^2}{2 \cdot (0.2 t)}\right) \leq \exp\left(-0.025 \cdot t\right) < 0.607.
\]
and
\[
\Pr[Z_{t} \geq w^+] \leq \exp\left(-\frac{(t - 0.25 t)^2}{(1 + 1/0.25) \cdot (0.25 t)}\right) \leq \exp\left(-0.45 \cdot t\right) < 0.001.
\]
Therefore, $\alg \geq \frac{1}{45} \cdot P \cdot (w^+ - w^-)$.

On the other hand, we have
\begin{align*}
\alg \geq \cs(\rank_n) \geq P \cdot \Pr[Z_{\lfloor n/2 \rfloor} \geq w^-] \cdot \frac{\lfloor n/2 \rfloor}{n} \cdot w^-.
\end{align*}

Notice that $\Pr[Y \geq w^-] \geq 0.75$ and as in Case (1), we apply Chernoff bounds as follows:
\begin{align*}
\Pr[Z_{\lfloor n/2 \rfloor} < w^-] \leq \exp\left(-\frac{(0.75n - 0.5n)^2}{2 \cdot (0.75n)}\right) \leq \exp\left(-0.041 \cdot n\right) < 0.01.
\end{align*}
Therefore, $\alg \geq \frac{1}{3} \cdot P \cdot w^-$ and further $\alg \geq \frac{1}{50} \cdot P \cdot w^+ \geq \frac{1}{250} \cdot \core$.

\paragraph{Case (4): $S(u^-) < 0.1$ and there is no value $w$ with $S(w) \in [2 \cdot S(u^-), 4 \cdot S(u^-)]$.}
Let $w^+$ be the smallest value with $S(w^+) < 2 \cdot S(u^-)$, and let $w^-$ be the largest value with $S(w^-) > 4 \cdot S(u^-)$. These values are consecutive in the support of $Y$.

We have $w^+ - w^- \geq \frac{1}{5} \cdot u^-$, since by \cref{lem:reg_concave},
\[
w^+ \cdot S(w^+) \geq \frac{1 - S(w^+)}{1 - S(u^-)} \cdot u^- \cdot S(u^-) \geq \frac{0.8}{0.9} \cdot u^- \cdot \frac{1}{2} \cdot S(w^+),
\]
and
\[
w^+ \cdot S(w^+) \geq w^- \cdot S(w^-) \geq w^- \cdot 2 \cdot S(w^+).
\]
where we have used $w^+ > w^-$. These imply $w^+ \ge \frac{4}{9} u^-$ and $w^- \le \frac{w^+}{2}$, so that $w^+ - w^- \geq \frac{1}{5} \cdot u^-$.

Choose $t \geq 20$ so that $\Pr[Y \geq w^+] \in [0.2 \cdot t/n, 0.25 \cdot t/n]$. This $t$ exists since $\Pr[Y \geq w^+] \geq \Pr[Y \geq u^-] \geq \frac{4}{n}$ and $\Pr[Y \geq w^+] \leq 0.2$. Similar to Case (2), this now gives
\begin{align*}
\alg \geq \cs(\rank_t) &\geq P \cdot \Pr[Z_{\lfloor t/10 \rfloor} \geq w^+ \land Z_{t} \leq w^-] \cdot \frac{\lfloor t/10 \rfloor}{t} \cdot (w^+ - w^-)\\
&\geq P \cdot (1 - \Pr[Z_{\lfloor t/10 \rfloor} < w^+] - \Pr[Z_{t} > w^-]) \cdot \frac{1}{75} \cdot u^-\\
&\geq \frac{1}{75} \cdot (1 - \Pr[Z_{\lfloor t/10 \rfloor} < w^+] - \Pr[Z_{t} \geq w^+]) \cdot \core.
\end{align*}

Now notice that $\Pr[Y \geq w^+] \geq 0.2 \cdot t/n$ and $\Pr[Y \geq w^+] < 0.25 \cdot t/n$. As in Case (1), we apply Chernoff bounds:
\[
\Pr[Z_{\lfloor t/10 \rfloor} < w] \leq \exp\left(-\frac{(0.2 t - 0.1 t)^2}{2 \cdot (0.2 t)}\right) \leq \exp\left(-0.025 \cdot t\right) < 0.607.
\]
and
\[
\Pr[Z_{t} > b] \leq \exp\left(-\frac{(t - 0.25 t)^2}{(1 + 1/0.25) \cdot (0.25 t)}\right) \leq \exp\left(-0.45 \cdot t\right) < 0.001.
\]
Therefore,
$\alg \geq \frac{1}{250} \cdot \core$.
\end{proof}

\subsubsection{Completing Proof of \cref{thm:regularity}}
Using \cref{lem:coretail}, \cref{lem:tail}, and \cref{lem:core}, we obtain:
\begin{align*}
 \cs_0 & \leq \core + \tail  = 5 \cdot \core + (\tail - 4 \cdot \core)  \leq 1650 \cdot \alg + 250 \cdot \alg  = 1900 \cdot \alg. \qedhere
\end{align*}
This completes the proof of \cref{thm:regularity} and hence, \cref{thm:cs0} and \cref{thm:main0}.

\section{Approximating Consumer Surplus: The non-i.i.d. Case}
\label{sec:mis}
We now extend the machinery developed in \cref{sec:iid} to show the following theorem, where we assume $\D_i$ are independent, but otherwise arbitrary distributions. We restate \cref{thm:main2} here.

\begin{theorem} 
\label{thm:main}
When the $\D_i$'s are arbitrary and independent,\footnote{We present the proof assuming $\D_i$ are regular. To extend the proof to arbitrary distributions, we simply replace the virtual value with ironed virtual value throughout.} there is a signaling scheme achieving an $O\big(\!\min\!\left(n \log n, \K^2\right)\!\big)$-approximation to $\opt$. Further, this signaling scheme has computation time polynomial in $n$ and $\K$.
\end{theorem}

Recall the decomposition of $\opt$ into Myerson's surplus $\cs(\D)$ and the non-allocation surplus $\cs_0$ from \cref{sec:iid}.  In the non-i.i.d. case, there is also a loss in surplus because because the highest value and highest virtual-value buyers are different. We call this the {\em mis-allocation surplus}.  We denote the expected loss due to mis-allocating the item to buyer $i$ as $\cs_i$, and the expected loss due to mis-allocating to a buyer with value $x_k$ as $\widehat{\cs}_k$. This gives us two ways of decomposing $\opt$:
\begin{equation}
\label{eq:opt} 
\opt = \cs(\D) + \cs_0 + \sum_{i=1}^n \cs_i \qquad \mbox{and} \qquad \opt = \cs(\D) + \cs_0 + \sum_{k=1}^{\K} \widehat{\cs}_k.
\end{equation}

Note that the approximation bounds in \cref{sec:iid} do not hold for $\cs_0$ in the non-i.i.d. setting. However, Lemma~\ref{lem:nonIIDBBMk} did not require $i.i.d.$ distributions in its proof and holds as is. We use this lemma to derive weaker approximation bounds for $\cs_0$  in the non-iid case by using \cref{alg:CS0_2}.

\subsection{Approximating Non-allocation Surplus $\cs_0$}
We show the following theorem:
\begin{theorem}
\label{thm:cs0-noniid}
For the non-i.i.d. setting, $\cs(\S_0)$ from \cref{alg:CS0_2} is an $O\left( \min(\log n, \K) \right)$ approximation to $\cs_0$.
\end{theorem}

\paragraph{$O(\log n)$-Approximation of $\cs_0$} We first show $\cs(\S_0)$ is an $O(\log n)$-approximation to $\cs_0$. We construct a signaling scheme denoted by $\S_0^2$ such that $\cs(\S_0^2) \le \cs(\S_0)$. We will then prove that consumer surplus of $\S_0^2$ is an $O(\log n)$-approximation to $\cs_0$. 
 
We assign a weight $w_j$ to each rank $j \in \{1,2, \ldots, n\}$ and run $\TB_t$ with probability proportional to these weights. Note that this choice of the $t$ does not depend on $\vec{v}$. Formally:
   
\begin{algorithm}[H]
	\DontPrintSemicolon
	$w_j \gets \frac{1}{j+1} $ \quad for $j \in \{1,2,\ldots, n-2\}$; \qquad $w_{n-1} \gets 1$; \qquad $w_n \gets 1$\; 
	Choose rank $t \in \{1,2,\ldots, n\}$  where  rank $j$ is chosen with probability proportional to $w_j$. \;
	\Return $\TB_t (\vec{v},\D)$ \;
	\caption{\textsc{$\S_0^2 (\vec{v}, \D )$}}\label{alg:CS0_1}
\end{algorithm}

\begin{lemma} 
	\label{lem: noniid0}
The consumer surplus of $\S_0^2$ (and hence $\S_0$) is an $O(\log n)$-approximation to $\cs_0$. 
\end{lemma} 
\begin{proof}
We have: $
\cs(\S_0^2) = \frac{\sum_{t=1}^{n} w_t \cdot \cs(\TB_t)}{\sum_{t=1}^n w_t}.$ Recall the definition of $P$ from Lemma~\ref{lem:z1}. 

We first use \cref{lem:nonIIDBBMk} to lower bound the numerator of the above formula. By performing a change of summation order, we have
\begin{align*}
& \sum_{t=1}^{n-1} w_t \cs(\TB_t) + w_n \cs(\TB_n) \\
 \ge & P \left(\left(\sum_{t = 1}^{n - 2} \E[Z_t] \cdot \sum_{j = t}^{n - 2} \frac{1}{j(j+1)}\right) - \left(\sum_{t = 2}^{n - 1} \E[Z_t] \cdot \frac{1}{t}\right) +  \left(\sum_{t = 1}^{n - 1} \E[Z_t] \cdot \frac{1}{n - 1}\right) - \E[Z_n] \right)+ \cs(\TB_n)\\
= &P \cdot (\E[Z_1] - \E[Z_n]) + \cs(\TB_n).
\end{align*}
It is easy to see that $\cs(\TB_n)  \geq P \cdot \E[Z_n] $. Using Lemma~\ref{lem:z1}, we have: 
\[
\sum_{i=1}^{n} w_i \cs(\TB_i)  \ge P \cdot \E[Z_1] = \cs_0.
\]
We next have: $\sum_{t=1}^n w_t = 1+\sum_{t = 1}^{n-1} \frac{1}{t} \leq 1 + \ln n$. Therefore, we have: 
\[
\cs(\S_0^2) = \frac{\sum_{t=1}^{n} w_t \cdot \cs(\TB_t)}{\sum_{t=1}^n w_t} \ge \frac{\cs_0}{1+\ln n}.
\]
Since the scheme $\S_0$ chooses the $\TB_t$ with largest value, this completes the proof.
\end{proof}

\paragraph{$O(\K)$-Approximation of $\cs_0$.} For proving this part, we will assume $n >5$. For smaller values of $n$, the previous analysis already yields a constant-approximation to $\cs_0$.  Using Lemma~\ref{lem:z1}, we know $\cs_0 = P \cdot \E[Z_1]$. Setting $x_0=0$,  we therefore have: 
\[
\E[Z_1] = \sum_{k=1}^ \K \Pr [Z_1 \ge x_k] \cdot (x_k-x_{k-1}).
\]

Therefore, there is a $k^* \in \{1,2, \ldots, \K\}$ such that $\E[Z_1] \le \K \cdot \Pr [Z_1 \ge x_{k^*}] \cdot (x_{k^*}-x_{k^*-1})$, i.e.,
\[
\cs_0 \le \K \cdot P \cdot \Pr [Z_1 \ge x_{k^*}] \cdot (x_{k^*}-x_{k^*-1}).
\]

Therefore, if we show that $\cs(\S_0) \ge \Omega(1) \cdot P \cdot \Pr [Z_1 \ge x_{k^*}] \cdot (x_{k^*}-x_{k^*-1})$, then we have an $O(\K)$-approximation to $\cs_0$. We now need the following probability lemma: 

\begin{lemma} 
	\label{lem:coin}
Given $n$ independent Bernoulli random variables, $X_i \in \{0,1\}$ for $i \in \{1,2,\ldots,n\}$, let $N = \sum_i X_i$. Then there exists a value $j \in \{1,2, \ldots, n\}$ such that:
\[
\frac{1}{j} \left(\sum_{i = 1}^j \Pr[N \geq i]\right) - \Pr[N \ge j+1] \ge  \frac{\Pr[N \ge 1]}{15}.
\]
\end{lemma}
\begin{proof}
Notice that $\Pr[N \geq 2i] = \Pr[N \geq i] \cdot \Pr[N \geq 2i \mid N \geq i] \leq \Pr[N \geq i]^2$ for any integer $i$. There is some $j' \in \{0, 1, \ldots, n\}$, so that $\Pr[N \geq j] \geq 0.5$, while $\Pr[N \geq j + 1] \leq 0.5$. Fix any such $j$. 
\begin{itemize} 
\item If $j' = 0$, then $\Pr[N \geq 1] \leq 0.5$, and for $j=1$, we have:
\begin{align*}
\Pr[N \geq 1] - \Pr[N \geq 2] \geq \Pr[N \geq 1] - \Pr[N \geq 1]^2 \geq 0.5 \cdot \Pr[N \ge 1].
\end{align*}
\item If $j' = 1$ and $\Pr[N \geq 1] \leq 0.75$, then for $j = 1$, we have:
\begin{align*}
\Pr[N \geq 1] - \Pr[N \geq 2] \geq \Pr[N \geq 1] - \Pr[N \geq 1]^2 \geq 0.25 \cdot \Pr[N \ge 1].
\end{align*}
\item If $j' = 1$ and $\Pr[N \geq 1] > 0.75$, then for $j = 1$, we have:
\begin{align*}
\Pr[N \geq 1] - \Pr[N \geq 2] \geq 0.25 \geq 0.25 \cdot \Pr[N \ge 1].
\end{align*}
\item If $j' \geq 2$, then for $j = j'$, we have
\begin{align*}
&\left(\frac{1}{j} \sum_{i = 1}^j \Pr[N \geq i]\right) - \Pr[N \geq j + 1]\\
\geq &\frac{1}{j}\left(\lfloor j/2 \rfloor \cdot \Pr[N \geq \lfloor j/2 \rfloor] + \lceil j/2 \rceil \Pr[N \geq j + 1]\right) - \Pr[N \geq j + 1]\\
\geq &\frac{\lfloor j/2 \rfloor}{j} \cdot \left(\Pr[N \geq \lfloor j/2 \rfloor] - \Pr[N \geq j + 1]\right)\\
\geq &\frac{\lfloor j/2 \rfloor}{j} \cdot \left(\sqrt{0.5} - 0.5\right) \geq \frac{1}{15} \geq \frac{1}{15} \cdot \Pr[N \geq 1]. \qedhere
\end{align*}
\end{itemize}
\end{proof}

We now complete the proof of the $O(\K)$-approximation in the lemma below.
\begin{lemma} 
	\label{lem:supk0}
	There exists a value $t'$ such that $\cs(\TB_{t'}) \ge \Omega(1) \cdot P \cdot \Pr [Z_1 \ge x_{k^*}] \cdot (x_{k^*}-x_{k^*-1})$. 
\end{lemma}
\begin{proof}
Recall that $Y_i$ is the distribution of $\D_i$ conditioned on being strictly below the reserve price $r_{\D_i}$. Define $G_k$ as the number of $Y_i$ with value at least $x_{k^*}$. Using Lemma~\ref{lem:nonIIDBBMk}, we have: 
\begin{align*}
\cs(\TB_t) \ge P  \E\left[\frac{\sum_{i=1}^t Z_i}{t} - Z_{t+1} \right]
&\ge P (x_{k^*} - x_{k^*-1}) \left(\frac{1}{t}\left(\sum_{i = 1}^t \Pr[Z_i \geq x_{k^*}]\right) - \Pr[Z_{t + 1} \geq x_{k^*}]\right)\\
&\ge P (x_{k^*} - x_{k^*-1}) \left(\frac{1}{t}\left(\sum_{i = 1}^t \Pr[G_{k^*} \geq i]\right) - \Pr[G_{k^*} \geq t + 1]\right).
\end{align*}
	
We define a Bernoulli random variable $X_i$ that is $1$ when $Y_i \ge x_{k^*}$ and zero otherwise. Applying Lemma~\ref{lem:coin}, there exists a value of $t'$ between $1$ and $n$ such that the following holds: 
\[
\left(\frac{1}{t'}\left(\sum_{i = 1}^{t'} \Pr[G_{k^*} \geq i]\right) - \Pr[G_{k^*} \ge t'+1]\right) = \Omega(1) \cdot \Pr[G_{k^*} \ge 1] = \Omega(1) \cdot \Pr [Z_1 \ge x_{k^*}].
\]
This when combined with the previous inequality, completes the proof. 
\end{proof}

From the previous lemma, there is a $t'$ such that consumer surplus of $\TB_{t'}$ is an $O(\K)$-approximation and $\cs(\S_0) \ge \cs(\TB_{t'})$ for any $t'$. Therefore, we have the following corollary: 
\begin{corollary}
\label{cor2}
The consumer surplus of $\S_0$ is an $O(\K)$-approximation to $\cs_0$.
\end{corollary}

\subsection{Approximating Mis-allocation Surplus}
To prove Theorem~\ref{thm:main}, Eq.~(\ref{eq:opt}) shows that this requires approximating $\cs_i$ and $\widehat{\cs}_k$ for any $i$ and $k$.  We now how the machinery used for approximating $\cs_0$ (\cref{alg:CS0_1}) can also be used for approximating each of the terms $\cs_i$ and $\widehat{\cs}_k$ above. The key idea is to apply \cref{alg:CS0_2} to a suitably truncated distribution of a carefully chosen buyer. This will yield the proof of \cref{thm:main}. 

We now derive an expression for $\cs_i$, the surplus lost due to allocating the item to buyer $i$. Assuming we break ties in favor of the higher valued buyer, we have:
\begin{equation}
\label{eq:csi}
\cs_{i} = \sum_{\substack{\vec{v}: i = \mbox{argmax}_{j} (\varphi_{\D_j}(v_j)) }} \Pr[\D = \vec{v}] \cdot \left(\max_j (v_j) - v_i\right).
\end{equation}
Similarly,  $\widehat{\cs}_k$ is the surplus lost when the item is allocated to a buyer with value $x_k$. We have:
\begin{equation}
\label{eq:csk}
\widehat{\cs}_{k} = \sum_{\substack{\vec{v}: x_k = \mbox{max}_{j} (\varphi_{\D_j}(v_j)) }} \Pr[\D = \vec{v}] \cdot \left(\max_j (v_j) - x_k\right).
\end{equation}

To approximate these quantities, we run the signaling scheme for approximating $\cs_0$ (\cref{alg:CS0_2}) on a {\em modified} product distribution.  In this new scheme, we fix a cut-off value $c$, and reveal the identity of all the buyers with realized value strictly greater than $c$. Let $a$ denote the largest realized value that is at most $c$, and $T$ denote the set of buyers with values strictly bigger than $a$.  We modify the distribution $\D_i$ for $i \in T$ as follows: Recall that $\D_{i | > c}$ denotes the distribution of $V_i$ conditioned on $V_i>c$; let $X_i$ denote the corresponding random variable. We change the distribution to be that of $X_i - a$, that we denote $\D_{j | >c} -a$. We now subtract $a$ from all the valuations $v_i, i \in T$, and run the signaling scheme $\S_0$ (\cref{alg:CS0_2}) on this instance.  The details of the scheme can be found in Algorithm~\ref{alg:cut}, where $c$ is the cutoff parameter, and where we denote by $\vec{v}_{T}$ the $|T|$ dimensional vector made by choosing the indices in $T$ from $\vec{v}$. Note that $c$ could be different from $a$ when there is no buyer whose value coincides with $c$.

\begin{algorithm}[H]
	\DontPrintSemicolon
	$ a \gets \max_{v_i \le c} (v_i)$; and $T\gets \{j : v_j > a \}$  \; 
	$\D_{T,a} \gets \prod_{j\in T_i} (\D_{j | >c} -a)$ \tcp*{Modified distributions for $j \in T$.}
	$\vec{v'} \gets  \vec{v}_{T} - \mathbf{a}$ \tcp*{$\mathbf{a}$ is a vector with all elements equal to $a$.}
	$ s \gets \S_0(\vec{v'}, \D_{T,a})$   \tcp*{$s$ is the signal returned by $\S_0$.}
	\Return $(a, T_i, s)$ as final signal
	\caption{\textsc{$\cut(\vec{v}, \D, c)$}}\label{alg:cut}
\end{algorithm}

\subsubsection{$O(n \log n)$-Approximation to $\opt$}
We choose $i^* \in \{1,2,\ldots,n\}$ that maximizes $\cs_i$. By~\cref{eq:opt}, to get an $O(n \log n)$-approximation to $\opt$, it suffices to demonstrate an $O(\log n)$-approximation to the quantity $\cs_{i^*}$.  The scheme $\S_1$ chooses $i^* \in \{1,2,\ldots, n\}$ that maximizes $\cs_{i}$ and returns $\cut(\vec{v},\D, v_{i^*})$. 

\begin{theorem}
\label{lem:csilog}
The consumer surplus of $\S_1$ is an $O(\log n)$-approximation to $\cs_{i^*}$.	
\end{theorem}
\begin{proof}
Applying \cref{lem: noniid0}  to the modified distribution $\D_{T,a}$ defined in~\cref{alg:cut}, the consumer surplus of $\S_1$ is at least: 
\[
\frac{1}{1 + \ln n} \sum_{T,a} \Pr(T,a) \cs_0(\D_{T,a})
\]
where $\Pr(T,a)$ is the probability of the event that $v_{i^*} = a$ {\em and}  the set of all the buyers with value strictly larger than $a$ is $T$. In order to prove the lemma, we need to show the following: 
\[
\cs_{i^*} \le \sum_{T,a} \Pr(T,a) \cs_0(\D_{T,a}).
\]
From Eq.~(\ref{eq:csi}), we have:
\[
\cs_{i^*} =  \sum_{\substack{\vec{v}: i^* = \mbox{argmax}_{j} (\varphi_{\D_j}(v_j)) }} \Pr[\D = \vec{v}] \cdot \left(\max_j (v_j) - v_{i^*}\right).
\]
Let $\D'_j =  \D_{j | > a} - a$ and let its reserve price be $r'_j$. Then,
\[
\Pr(T,a) \cs_0(\D_{T,a}) = \sum_{\vec{v'}: v'_j < r'_j \forall j \in T} \Pr(T,a) \cdot \Pr[\D_{T,a} = \vec{v'}] \cdot \max_{j \in T} (v'_j).
\]

Note that for every vector $\vec{v} \in \D$ there is a corresponding set $T$, value $a$, and vector $\vec{v'} \in D_{T,a}$ as defined in Algorithm~\ref{alg:cut}. We show that if a positive value is added to $\cs_{i^*}$ in the above formula, at least the same amount will be added to the corresponding $\Pr(T,a) \cs_0(\D_{T,a})$ by the corresponding vector $v'$. Since $\D_{T,a}$ is the product of conditional distributions of values of buyers in $T$ being larger than $a$, we have:  
\[
\Pr [\D = \vec{v}] = \Pr(T,a) \cdot \Pr [\D_{T,a} = \vec{v'}].
\]
Now, there is a positive contribution to $\cs_{i^*}$ in the event where $v_j = \max_i v_i$ and $ \varphi_{\D_j}(v_j) < \varphi_{\D_{i^*}}(v_{i^*}) \le v_{i^*} = a$.  This assumes ties are broken by allocating to the buyer with higher valuation. The contribution to $\cs_{i^*}$ is $v_j - v_{i^*}$.

Consider the distribution $\D_j$. By Lemma~\ref{lem:TRres},  $\varphi_{\D'_j}(v_j - a) = \varphi_{\D_j}(v_j)$. Since $\varphi_{\D_j}(v_j) < a = v_{i^*}$, we have $\varphi_{\D'_j}(v'_j) < 0$ so that $v'_j < r'_j$. Therefore, the contribution to $\cs_0(\D_{T,a})$ is precisely $v'_j = v_j - v_{i^*}$. This completes the proof.
\end{proof}

Since there are $n$ possible choices of $i^*$, this directly implies:

\begin{corollary}
\label{cor3}
If the values of the $n$ buyers are independent, there is a signaling scheme that is an $O(n \log n)$-approximation to $\opt$.
\end{corollary}

\subsubsection{$O(\K^2)$-Approximation to $\opt$.}
Similar to~\cref{app:cs0}, consider $k^* \in \{1,2,\ldots, \K\}$ that maximizes $\widehat{\cs}_k$. The scheme, denoted by $\S_2$ executes $\cut(\vec{v},\D, x_{k^*})$. Note that there may be no buyer with valuation $x_{k^*}$,  but~\cref{alg:cut} handles this case as well. 

\begin{theorem}
\label{lem:csklog}
The consumer surplus of $\S_2$ is an $O(\K)$-approximation to $\widehat{\cs}_{k^*}$. 
\end{theorem}
\begin{proof}
The proof is similar to the proof of Theorem~\ref{lem:csilog} and we omit details. By applying Corollary~\ref{cor2}  to the modified distribution $\D_{T,x_{k^*}}$, the consumer surplus of $\S_2$ is at least: 
\[
\Omega\left(\frac{1}{\K} \right) \sum_{T} \Pr(T) \cs_0(\D_{T,x_{k^*}})
\]
where $\Pr(T)$ is the probability that the set of all the buyers with value strictly larger than $x_{k^*}$ is $T$.  Furthermore using Eq.~(\ref{eq:csk}) and by replacing $v_{i^*}$ by $x_{k^*}$ in the proof in Theorem~\ref{lem:csilog}, we have 
\[
\widehat{\cs}_{k^*} \le \sum_{T} \Pr(T) \cs_0(\D_{T,x_{k^*}}).
\]
Combining the above two bounds proves the theorem.
\end{proof}

Since there are $\K$ possible values of $k^*$, the previous theorem directly implies:
\begin{corollary} 
\label{cor4}
If the values of the buyers are independent and have a common support of size $\K$, there is a signaling scheme that is an $O(\K^2)$-approximation to $\opt$.
\end{corollary}

Combining~\cref{cor3,cor4} completes the proof of~\cref{thm:main}.

\section{Conclusion and Open Questions}
Note that our $\TB_t$ mechanism can be viewed as a screening procedure -- the intermediary only allows a fixed number of high-value bidders to bid. When the intermediary is an independent (typically, governmental) agency, such screening would map to ``pre-certifying'' bidders entering into private auctions. Similarly, when real-estate agencies have agents representing both sellers and buyers, they could (and often do) recommend a particular listing only to a chosen set of buyers based on better knowing their utilities. Therefore, as a side-effect, our procedures yield realistic mechanisms for an intermediary to increase surplus for both buyers and the seller. 

In terms of open questions, beyond improving the lower and upper bounds in our specific setting (both existence and computational), it would be interesting to explore the equilibria in optimal auctions when the intermediary can send different signals to the seller and to the buyers, much like in~\citep{bbm2017first,shen2019buyer}. At an even higher level, our work can be considered a special case of a larger problem of information intermediaries for multi-agent mechanisms. As mentioned before, in our case, the optimal auction is the mechanism, and the intermediary can change the information to this mechanism in order to achieve ``fairness'' between producer and consumer surplus. It would be interesting to explore the question of achieving fairness by selectively regulating information to a black-box optimizer or mechanism in more general settings.

\paragraph{Acknowledgment.} This work is supported by NSF grants ECCS-1847393, DMS-1839346, CNS-1955997, and CCF-2113798, and and ONR award N00014-19-1-2268.  Reza Alijani is now at Google. 


\newpage

\bibliographystyle{plainnat}
\bibliography{ref}

\end{document}